\documentclass[twocolumn,a4paper,showpacs,superscriptaddress,pra,eqsecnum,nofootinbib]{revtex4}

\usepackage{amsfonts,amsmath,graphicx,color,epsfig,amssymb}
\usepackage{footmisc}
\DeclareGraphicsExtensions{.jpeg,.jpg,.ps,.pdf,.pdf}
\newcommand{\nn}{\nonumber}
\newcommand{\beq}{\begin{equation}}
\newcommand{\eeq}{\end{equation}}
\newcommand{\bqa}{\begin{eqnarray}}
\newcommand{\eqa}{\end{eqnarray}}
\newcommand{\erf}[1]{Eq.~(\ref{#1})}
\newcommand{\erfs}[2]{Eqs.~(\ref{#1}) and (\ref{#2})}
\newcommand{\erflist}[2]{Eqs.~(\ref{#1}) -- (\ref{#2})}
\newcommand{\dg}{^\dagger}

\newcommand{\bra}[1]{\langle{#1}|}
\newcommand{\ket}[1]{|{#1}\rangle}
\newcommand{\sch}{Schr\"odinger}

\newcommand{\cu}[1]{\left\{{#1} \right\}}

\newcommand{\tp}{^{\top}}
\newcommand{\s}[1]{\hat\sigma_{#1}}

\newtheorem{theorem}{Theorem}

\newtheorem{lemma}[theorem]{Lemma}

\newenvironment{proof}[1][Proof]{\noindent\textbf{#1.} }{\ \rule{0.5em}{0.5em}}
 \newcommand{\ea}{{\em et al.}}
\begin{document}

\title{Entanglement, EPR-correlations, Bell-nonlocality, and Steering}
 \author{S. J. Jones}
\affiliation{Centre for Quantum Computer Technology, Centre for
  Quantum Dynamics, Griffith University, Brisbane,
  4111 Australia}
 
\author{H. M. Wiseman}
\affiliation{Centre for Quantum Computer Technology, Centre for
  Quantum Dynamics, Griffith University, Brisbane,
  4111 Australia}  
  
  \author{A.\ C.\ Doherty} 
\affiliation{School of Physical Sciences,
  University of Queensland, Brisbane 4072 Australia}
  
 \newcommand{\red}{\color{red}}
\newcommand{\blk}{\color{black}}
  
  \date{\today}

\begin{abstract}
In a recent work [Phys. Rev. Lett. {\bf 98}, 140402 (2007)] we defined ``steering'', a type of quantum nonlocality that is logically distinct from both nonseparability and Bell-nonlocality. In the bipartite setting, it hinges on the question of whether Alice can affect Bob's state at a distance through her choice of measurement. More precisely and operationally, it hinges on the question of whether Alice, with classical communication, can convince Bob that they share an entangled state, under the circumstances that Bob trusts nothing that Alice says. We argue that if she can, then this demonstrates the nonlocal effect first identified in the famous EPR paper [Phys. Rev. {\bf 47}, 777 (1935)] as a universal effect for pure entangled states. This ability of Alice to remotely prepare Bob's state was  subsequently called steering by \sch, whose terminology we adopt. The phenomenon of steering has been largely overlooked, and prior to our work had not even been given a rigorous definition that is applicable to mixed states as well as pure states. Armed with our rigorous definition, we proved that steerable states are a strict subset of the entangled states, and a strict superset of the states that can exhibit Bell-nonlocality. In this work we expand on these results and provide further examples of steerable states.  We also elaborate on the connection with the original EPR paradox.
%
%The concept of steering was introduced by \sch\ in 1935 as a
%generalization of the EPR paradox for arbitrary pure bipartite
%entangled states and arbitrary measurements by one party. In a recent work [Phys. Rev. Lett. {\bf 98}, 140402 (2007)] we provided an operational definition of steering which showed that steerable states are a strict subset of the entangled states,
%and a strict superset of the states that can exhibit Bell-nonlocality (i.e. violation of a Bell inequality). In this work we expand 
%%expatiate 
%on these results and provide further examples of steerable states.  We also elaborate on the connection with the original EPR paradox.
   \end{abstract}

%\pacs{130.102.131.21}- Steering?

\maketitle

\section{Introduction}\label{Intro}

Entanglement is arguably the central concept in the field of quantum information. 
However, there is an unresolved tension between different notions of what entanglement is,
even in the bipartite setting. %(although see Ref.~\cite{MasEtalARX07}).
On the one hand, entangled states are defined as those that cannot be created from factorizable states
using local operations and classical communication (LOCC). On the other hand, entanglement
is regarded as a resource that enables the two parties to perform interesting or (in more recent times) 
useful nonlocal protocols. For pure states, which were the only states considered in this context 
for many decades, these notions coincide, and the word 
``entangled'' (introduced by \sch\ \cite{SchPCP35}) is identical with ``factorizable''.

The first authors to identify an interesting nonlocal effect associated with unfactorizable states
were Einstein, Podolsky and Rosen (EPR) in 1935 \cite{EinEtalPR35}. 
They considered a general unfactorizable pure state of two systems, held by two distant 
parties (say Alice and Bob)\footnote{All we have 
changed from EPR's presentation is to use Dirac's notation rather than wave functions.}:
\beq \label{psient} \ket{\Psi} =
\sum_{n=1}^{\infty} c_n\ket{u_n}\ket{\psi_n} = \sum_{n=1}^{\infty}
d_n\ket{v_n}\ket{\varphi_n}. \eeq 
Here $\cu{ \ket{u_n}}$ and $\cu{
\ket{v_n}}$ are two different orthonormal bases for Alice's system . 
If such states exist, then if 
Alice chose to measure in the $\cu{\ket{u_n}}$ (respectively $\cu{\ket{v_n}}$)
basis, then  she would instantaneously collapse Bob's system  into
one of the states $\ket{\psi_n}$ (respectively $\ket{\varphi_n}$). 
That is, ``as a consequence of two different measurements performed upon the first
system, the second system may be left in states with two different wavefunctions.'' \cite{EinEtalPR35}.
%If quantum mechanics (QM) gave a complete description of the world, 
%the one immediately sees a problem, which EPR identify:  
Now comes the paradox: ``the two systems no
longer interact, [so] no real change can take place in [Bob's]
system in consequence of anything that may be done to [Alice's]
system.''\cite{EinEtalPR35} That is, if Bob's quantum state is the real state of his system, 
then Alice  cannot choose to make it collapse into either one of the $\ket{\psi_n}$ 
or one of the $\ket{\phi_n}$ because that would violate 
local causality. Note that it is crucial to consider more than one sort of measurement for Alice; 
if Alice were restricted to measuring in one basis (say the $\ket{u_n}$ basis),
then it would be impossible to demonstrate any ``real change'' in Bob's system, because
she might know beforehand which of the $\ket{\psi_n}$ is the real state of his system.
 That is, the paradox exists only if there is not a local hidden state (LHS) model 
for Bob's system, in which the real state $\ket{\psi_n}$ is hidden from Bob but may be known
to Alice. 

As the above quotations show, EPR assumed local causality to be a true feature of the 
world; indeed, they say that no ``reasonable'' theory could be expected to permit otherwise.  
They thus concluded that the wavefunction cannot describe reality; that is,  the quantum mechanical (QM) description must be {\em incomplete}. 
%\footnote{Not, as Bohr \cite{Boh35} seemed to think, that it was {\em incorrect}.
%See Ref.~\cite{WisCP06} and references within.}.  
Their intuition was thus that local causality could be maintained by completing QM. 
This intuition was supported by the famous example that they then presented as a special case of \erf{psient}, involving
a bipartite entangled state with perfect correlations in position and momentum. 
The ``EPR paradox'' in this case is trivially resolved by considering 
local hidden variables (LHVs) for position and momentum.

 Although the argument of EPR against the completeness of QM was correct, 
their intuition was not.  As proven 
by Bell \cite{BelPHY64,Bel71}, local causality cannot be maintained even if one 
allows QM to be completed by hidden variables. That is, assuming as always 
(and with good justification \cite{HarCP98}) that 
QM is correct, Bell's theorem proves that local causality is not a true feature of the world\footnote{For both Bell and EPR, there is an escape, ``by denying independent real situations 
as such to things which are spatially separated from each other,'' as stated by Einstein in 1946 
\cite{Ein46}. That is, Alice, for example, can refuse to admit the reality of Bob's measurement results
until she observes them, by talking only about the outcomes of her own future measurements. However Einstein 
 stated that in his opinion this anti-realism was ``equally unacceptable'' as violating local causality; 
 see Ref.~\cite{WisCP06} for a discussion.}.
Interestingly, any unfactorizable pure state can be used not only to demonstrate the 
EPR paradox \cite{EinEtalPR35}, but also to demonstrate Bell-nonlocality (that is, the violation of local causality). 
(This fact was perhaps first stated in 1989 by Werner~\cite{WerPRA89}; the first detailed proof was given 
in 1991 by Gisin \cite{GisPLA91}; see also \cite{PopRohPLA92}.)

With the rise of quantum information {\em experiment}, 
the idealization of considering only pure states has become untenable. 
The question of which {\em mixed} states were Bell-nonlocal 
(that is, allowed a demonstration of Bell-nonlocality) was first addressed by Werner \cite{WerPRA89}, %in a foundational paper in the field of quantum information theory. 
in a foundational paper pregnant with implications for, and applications in, quantum information  science. Revealing the first hint of the complexity of mixed-state entanglement, 
still being uncovered \cite{MasEtalARX07},
Werner showed that not all mixed entangled states can demonstrate Bell-nonlocality.  Here,
for mixtures, an entangled state is defined as one which cannot be written as a mixture of factorizable pure states. Indeed Werner's paper is often cited as that which 
introduced this definition. That is, he is credited 
with introducing the dichotomy of entangled states versus separable (i.e. locally preparable) states. However it is interesting to note that he used neither the 
term entangled nor the term separable. For a discussion of the history of terms 
used in this context, and their relation to the present work, see Appendix~\ref{App:history}.

In a recent Letter, the present authors also considered the issue of mixed states and nonlocality \cite{WisEtalPRL07}.  We rigorously defined
the class of states that can be used to demonstrate the nonlocal effect which EPR identified in 1935. 
We proposed the term ``steerable''
for this class of states (for reasons given in Appendix~\ref{App:history}), 
and proved that the set of Bell-nonlocal states is a strict subset
of the set of steerable states, which in turn is a strict subset of the set of nonseparable states. This was our main result.

Like ``entangled'', ``steering'' is a term introduced by \sch\ \cite{SchPCP35} in the aftermath of the EPR paper. Specifically, he credits EPR with calling attention to ``the obvious but very disconcerting fact'' that for a pure entangled state like \erf{psient}, 
Bob's system can be 
``steered or piloted into one or the other type of state at [Alice's] mercy in spite of [her] having no access to it.'' He referred to this 
as a ``paradox'' \cite{SchPCP35,SchPCP36} because if such states can exist, and if the QM description is complete, then local causality must be violated.

%\bibitem{Pri84}
%[SEE Werner's first reference for this].

In Ref.~\cite{WisEtalPRL07} we first supplied an operational definition of steering in the style of 
a quantum informational task involving two parties (in contrast to demonstrating Bell-nonlocality, 
which can be defined as a task involving three parties). Next we turned this operational definition into a mathematical
definition. Applying this to the case of $2\times 2$ dimensional Werner states enabled us to establish our main result, quoted above. We then completely characterized steerability for $d\times d$-dimensional Werner states and isotropic states. Finally, we completely characterized the Gaussians states that are 
steerable by Gaussian measurements, and related this to the Reid criterion \cite{ReiPRA89} for the EPR paradox. 

In the present paper we expand and extend the material in Ref.~\cite{WisEtalPRL07}. 
In Sec.~II we present the operational definitions of Bell-nonlocality and steering as before, and
also that for demonstrating nonseparability. In addition we use these operational definitions to show that they lead to a hierarchy of states: Bell-nonlocal within steerable within nonseparable. 
In Sec.~III we turn our operational definitions into mathematical definitions, and in addition we 
explain how our definition of steering conforms to \sch's use of the term.
In Sec.~IV we derive  conditions for steerability for four 
families of states. As before, we consider
 Werner states, isotropic states and Gaussian states, but here we expand the proofs for the benefit of the reader.
 In addition, we consider another class of states: the ``inept states'' of Ref.~\cite{JonEtalPRA05}.  We also consider a 
 subclass of Gaussian states in more detail: the symmetric two-mode states produced in parametric down-conversion. 
 %We identify identities between certain Werner (or isotropic) states and certain inept states, and between certain 
% inept states and certain symmetric two-mode states. 
 We conclude with a summary and discussion in Sec.~V.

\section{Operational Definitions}
It is useful to begin with some operational definitions for the different properties of quantum states that we wish to consider.  This is useful for a number of reasons.  First, it presents the ideas that we wish to discuss in an accessible format for those familiar with concepts in modern quantum information.  Second, it allows us to present an elementary proof of the hierarchy of the concepts (we will present a more detailed proof of this hierarchy in subsequent sections).
%motivation for operational definitions. a) in the style of modern quantum information.
%b) allows elementary proof of hierarchy of concepts.

First, let us define the familiar concept of Bell-nonlocality
\cite{BelPHY64} as a task, in this case with three parties; Alice, Bob and Charlie. Alice
and Bob can prepare a shared bipartite state, and repeat this any
number of times. Each time, they measure their respective parts.
Except for the preparation step, communication between them is
forbidden (this prevents them from colluding in an attempt to fool Charlie). Their task is to convince Charlie (with whom they can
communicate) that  the state they can prepare is entangled. Charlie
accepts QM as correct, but trusts neither Alice nor Bob. If the
correlations between the results they report {\em can} be
explained by a LHV model, then Charlie will not be convinced that the state is entangled; the
results could have been fabricated from shared classical randomness. 
Conversely, if the correlations {\em cannot} be so explained then the state must be entangled.
Therefore they will succeed in their task \emph{iff} (if and only if) they
can demonstrate Bell-nonlocality. This task can be thus considered as an operational definition of violating a Bell inequality.

The analogous definition for steering uses a task with only two parties.
Alice can prepare a bipartite quantum state and
send one part to Bob, and repeat this any number of times. Each time,
 they measure their respective parts, and communicate classically.
 Alice's task is to convince Bob that the state she can prepare is
 entangled. Bob (like \sch) accepts that QM describes the results
of the measurements he makes (which, we assume, allow him to do
local state tomography). However Bob does not trust Alice. 
 In this case Bob must determine whether the correlations between his 
local state and Alice's reported results are proof of entanglement. 
 How he should determine this is explained in detail in Sec.~III,
but the basic idea is that he should not accept the correlations as proof of
entanglement if they can be explained by a LHS model for Bob. 
If the
correlations between Bob's measurement results and the results Alice
reports {\em can} be so  explained then Alice's results could have been fabricated from
her knowledge of Bob's LHS in each run. Conversely, if the correlations
{\em cannot} be so explained then the
bipartite state must be entangled. Therefore we say that Alice
will succeed in her task \emph{iff} she can steer Bob's state.

Finally the simplest task is for Alice and Bob to determine if a bipartite quantum state that they share is nonseparable. In this case they can communicate results to one another, they trust each other, and they can repeat the experiment sufficiently many times to perform state tomography. By analyzing the reconstructed bipartite state, they could determine whether it is nonseparable. That is, whether it can be described by correlated LHSs for Alice and Bob.  Because Alice and Bob trust each other and can freely communicate, this is really a \emph{one} party task.

Using these operational definitions we can show that Bell nonlocality is a stronger concept than steerability. That is, that Bell nonlocal states are a subset of the steerable states. 
%Based on these operational definitions one may wonder which concept specifies the stronger condition.  By which we mean, is one of these classes of states a subset of the other? One may intuitively expect Bell nonlocality to be a stronger concept than steerability; that is, that Bell nonlocal states are a subset of steerable states.  This is indeed the case, as can be seen as follows.  
The operational definition of Bell nonlocality is based on three parties and requires a completely distrustful Charlie.  If we \emph{weaken} %\footnote{It is somewhat subjective to say that the condition is weakened by allowing one party to be trusted.  However, in a rigorous formulation of operational conditions one would expect that no trust be necessary.  A discerning experimenter could test the condition based on the information available to them with no caveats about trusting its origins. With this in mind, reformulating the condition by including the necessity of trusting one party can reasonably be claimed as a weaker condition.} 
this condition by allowing Charlie to trust Bob completely, we arrive at the following situation.  Charlie can now in principle do state tomography for Bob's local state (as he believes everything told to him by Bob), and he only distrusts the measurement results reported by Alice.  In this case, he will only concede that the state prepared by Alice and Bob is entangled if the state is steerable. Thus it is possible to arrive at the operational definition for steering by weakening the operational definition for Bell nonlocality. Thus, the Bell nonlocal states are a subset of the steerable states. %By doing so, we allow a wider class of states to satisfy the condition and we see immediately that the steerable states are a superset of Bell nonlocal states (that is, Bell nonlocal states are a subset of steerable states).

Similarly, if we weaken the condition for steerability we arrive at the condition for nonseparability as follows.  In this case we weaken the condition by allowing for Bob to trust Alice completely.  Since Bob now has access to the measurement information for both subsystems (as he believes everything told to him by Alice) he can in principle perform state tomography. Clearly in this situation Bob will only concede that they share an entangled state if the state that Alice prepares really is entangled. %Thus it is possible to arrive at the operational definition for nonseparability by weakening the operational definition for steerability.
Thus, the steerable  states are a subset of the entangled states. We illustrate these relations graphically in Fig. 1.

\begin{figure}
\begin{center}
\includegraphics[width=8.5cm]{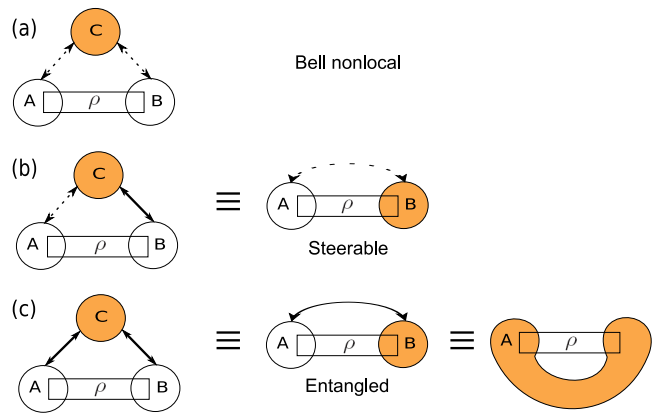}%.pdf
\end{center}
%\vspace{-5ex}
\caption{Operational definitions for classes of entangled states.  Bell nonlocal states (a) can be defined via a three-party task (involving Alice [A], Bob [B] and Charlie [C]). Steerable states (b) may be defined using a two-party task. Defining an entangled state (c) essentially requires only one party. In all cases shading indicates the skeptical party, dotted arrows indicate two-way communication, and solid arrows indicate trust and two-way communication.}
\label{Tasks}\end{figure}

While these operational definitions give a good insight into the relationships between the three classes of states it is also desirable to have a strict mathematical way to define the classes.  We present such definitions in the following section.

%For consistency we will begin with the three party case, but we now allow Charlie to trust \emph{both} Bob and Alice (this is equivalent to the two party case and allowing Alice and Bob to trust each other - then either party could verify the entanglement in the shared state).  Now Charlie has access to the measurement information for both subsystems (as he believes everything told to him by Alice and Bob) and can in principle perform state tomography. Clearly in this situation Charlie will only concede that Alice and Bob share an entangled state if the state they prepare really is entangled. Thus it is possible to arrive at the operational definition for nonseparability by weakening the operational definition for steerability. Thus, the steerable  states are a subset of the entangled states.

\section{Mathematical Definitions}
First, we define some terms.
Let the set of all observables on the
Hilbert space for Alice's system be denoted ${\mathfrak D}_\alpha$.
We denote an element of ${\mathfrak D}_\alpha$ by $\hat{A}$, and
the set of eigenvalues $\cu{a}$ of $\hat{A}$ by $\lambda(\hat A)$.
By $P(a|\hat{A};W)$ we mean the probability that Alice will obtain
the result $a$ when she measures $\hat{A}$ on a system with state
matrix $W$. We denote the measurements that Alice is able to perform by
the set  ${\mathfrak M}_\alpha \subseteq {\mathfrak D}_\alpha$. Note that, following Werner \cite{WerPRA89}, we are restricting to
projective measurements.  The corresponding notations
for Bob, and for Alice and Bob jointly, are obvious. Thus, for example, \beq
P(a,b|\hat{A},\hat{B};W) = {\rm Tr}[(\hat\Pi_a^A\otimes
\hat\Pi_b^B) W], \eeq where $\hat\Pi_a^A$ is the
projector satisfying $\hat{A}\hat\Pi_a^A =
a\hat \Pi_a^A$.

The strongest sort of nonlocality in QM is
Bell-non\-locality \cite{BelPHY64}. This is a property of entangled states which violate a Bell inequality.  This is exhibited in
an experiment on state $W$ \emph{iff} the correlations between $a$ and $b$
cannot be explained by a LHV model. That is, if it is
{\em not} the case that for all $a\in \lambda(\hat A), b \in
\lambda(\hat B)$,  for all $\hat{A} \in {\mathfrak M}_\alpha,
\hat{B} \in {\mathfrak M}_\beta$, we have \beq \label{Bell-local1}
P(a,b| \hat A,\hat B;W) = \sum_\xi \wp(a|\hat A,\xi) \wp (b|\hat
B,\xi) \wp_\xi. \eeq Here, and below, $\wp(a|\hat A,\xi)$,
$\wp(b|\hat B,\xi)$ and $\wp_\xi$ denote some (positive,
normalized) probability distributions, involving the LHV $\xi$. We
say that a {\em state} is Bell-nonlocal \emph{iff} there exists a measurement set
${\mathfrak M}_\alpha\times{\mathfrak M}_\beta$ that allows Bell-nonlocality
to be demonstrated. If \erf{Bell-local1} is always satisfied we say $W$ is Bell-local.

A strictly weaker \cite{WerPRA89} concept is that of %quantum
nonseparability or
entanglement. A nonseparable state is one that
{\em cannot} be written as \beq W = \sum_\xi
{\sigma}_\xi \otimes \rho_\xi \,\wp_\xi.\eeq Here, and below,
$\sigma_\xi \in {\mathfrak D}_\alpha$ and $\rho_\xi \in  {\mathfrak D}_\beta$ are
some (positive, normalized) quantum states. We can also give an operational
definition, by allowing 
Alice and Bob the ability to measure a quorum
of local observables, so that they can
reconstruct the state $W$ by tomography  \cite{DArEtalJPA01}. Since the
complete set of observables ${\mathfrak D}$ is obviously a quorum, we
can say that a
state $W$ is nonseparable %exhibits nonseparability
\emph{iff} it is {\em
not} the case that for all $a\in \lambda(\hat A), b \in
\lambda(\hat B)$, for all $\hat{A} \in {\mathfrak D}_\alpha,
\hat{B} \in {\mathfrak D}_\beta$, we have \beq \label{separable1}
P(a,b| \hat A,\hat B; W) = \sum_\xi P(a|\hat{A}; \sigma_\xi)
P(b|\hat{B};\rho_\xi) \wp_\xi. \eeq

Bell-nonlocality and nonseparability are both concepts that are
symmetric between Alice and Bob. However {\em steering}, \sch's term
for the 
EPR effect \cite{SchPCP35}, is inherently asymmetric.
It is about whether Alice, by her choice of measurement $\hat{A}$, can
collapse Bob's system into different types of states in the
different ensembles $E^A \equiv \cu{\tilde\rho^A_a:a\in
\lambda(\hat{A})}$. Here $\tilde{\rho}^A_a \equiv {\rm Tr}_\alpha[W (\hat\Pi_a^A\otimes {\bf I})] \in {\mathfrak D}_\beta$, is Bob's state conditioned on Alice measuring $\hat{A}$ with result $a$. The tilde denotes that this state is unnormalized (its norm is
the probability of its realization). Of course Alice
cannot affect Bob's unconditioned state
$\rho = {\rm Tr}_\alpha[W] = \sum_a \tilde \rho^A_a$
--- that would allow super-luminal signaling.
Despite this,
steering is clearly nonlocal if one believes that the state of a quantum 
system is a physical property of the system, as did \sch. This is apparent from his statement that
``It is rather discomforting that the theory should allow a system to be steered or piloted into one or the other type of state at the experimenter's mercy in spite of his having no access to it.''

As this quote also shows, \sch\ was not wedded to the terminology ``steering''. He also used the term  ``control''  for this phenomenon \cite{SchPCP36}, and the word ``driving'' in the context of his 1936 result that ``\ldots a sophisticated experimenter can \ldots produce a non-vanishing probability of driving the system into any state he chooses'' \cite{SchPCP36}. (By this he means that if a bipartite system is in a pure entangled state, then one party (Alice) can, by making a suitable measurement on her subsystem, create   any pure quantum state $\ket{\psi}$ for Bob's subsystem with probability $\bra{\psi}\rho^{-1}\ket{\psi}^{-1}$, whenever this is well-defined \cite{SchPCP36}.)   He regarded steering or driving as a ``necessary and indispensable feature'' of quantum mechanics \cite{SchPCP36}, but found it ``repugnant'', and doubted whether it was really true. That is, he was ``not satisfied about  there being enough experimental evidence for [its existence in Nature]'' \cite{SchPCP36}.

What experimental evidence would have convinced \sch? The pure entangled states he discussed are an idealization, so we cannot expect ever to observe precisely the phenomenon he introduced. On the other hand, \sch\  was  quite explicit that a separable but correlated state which allows ``determining the state of the first system by {\em suitable} [his emphasis] measurement of the second or {\em vice versa}'' could never exhibit steering. Of this situation, he says that  ``it would utterly eliminate the experimenter's  influence on the state of that system which he does not touch.''  Thus it is apparent that by steering \sch\ meant something that could not be explained by Alice simply finding out which state Bob's system is in, out of some predefined ensemble of states.  In other words, the  ``experimental evidence''   \sch\ sought is precisely the evidence that would
convince Bob that Alice has prepared an entangled state under the
conditions described in our first (operational) definition of
steering.

To reiterate, we assume that the experiment can be
repeated at will, and that Bob can do state tomography. Prior to all
experiments, Bob demands that Alice announce the possible ensembles
$\cu{E^A:\hat{A} \in {\mathfrak M}_\alpha}$ she can steer Bob's state into. In
any given run (after he has received his state), Bob should
randomly pick an ensemble $E^A$, and ask Alice to prepare it\footnote{This ensures that Bob need not trust Alice that they share the same
state $W$ in each run, because Alice gains nothing by preparing
different states in different runs, because she
never knows what ensemble Bob is going to ask for.}. 
Alice should then do so, by measuring $\hat{A}$ on her
system, and announce to Bob the particular member
$\rho^A_a$ she has prepared. Over many runs, Bob can verify that
each state announced is indeed produced, and is announced
with the correct frequency ${\rm Tr}[\tilde \rho^A_a]$.

If Bob's system did have a pre-existing LHS $\rho_\xi$, (as \sch\ thought),
then Alice could attempt to fool Bob, using her knowledge of $\xi$.
This state would be drawn at random from some prior ensemble of LHSs $F =
\cu{\wp_\xi\rho_\xi}$ with $\rho = \sum_\xi
\wp_\xi\rho_\xi$. Alice would then have to announce a LHS $\tilde{\rho}_a^A$
based on her knowledge of $\xi$,
according to some stochastic map from $\xi$ to $a$. %$\wp(a|A,\xi)$.
Alice will have failed to convince Bob that she can steer his system if, for all  $\hat{A} \in {\mathfrak M}_\alpha$, and for all
$a\in \lambda(\hat{A})$, there exists an ensemble $F$ and a stochastic map
$\wp(a|\hat{A},\xi)$
from $\xi$ to $a$
such that
\beq \label{steering2}
\tilde{\rho}_a^A = \sum_\xi  \wp(a|\hat{A},\xi) \rho_\xi  \wp_\xi.
\eeq
That is, if there exists a  {\em coarse-graining} of ensemble $F$ to ensemble $E^A$ then Alice may simply know Bob's pre-existing state $\rho_\xi$.
Conversely,  if Bob {\em cannot} find any ensemble $F$ and map $\wp(a|\hat{A},\xi)$ satisfying \erf{steering2} then Bob must admit that Alice can steer his system.

%If Bob is less skeptical and trusts Alice to choose
%measurements at random, then
We can recast this definition as a `hybrid' of \erfs{Bell-local1}{separable1}: Alice's
 measurement strategy ${\mathfrak M}_\alpha$ on state $W$ exhibits steering \emph{iff}
 it is {\em not} the case that
for all $a\in \lambda(\hat A), b \in \lambda(\hat B)$,  for all $\hat{A} \in {\mathfrak M}_\alpha, \hat{B} \in {\mathfrak D}_\beta$,  we can write
\beq P(a,b| \hat A,\hat B; W) = \sum_\xi
\wp(a|\hat{A},\xi)P(b|\hat{B};\rho_\xi) \wp_\xi.
\label{Steer}\eeq That is, if the joint probabilities for Alice and Bob's measurements can be explained using a LHS model for Bob and a LHV model for Alice correlated with this state, then we have failed to demonstrate steering.
\emph{Iff} there exists a measurement
strategy ${\mathfrak M}_\alpha$ that exhibits steering, we say that the state
$W$ is {\em steerable} (by Alice).

It is straightforward to see that the condition for no steering implies the condition for Bell-locality, since if there is a
model with $P(b|\hat{B},\rho_\xi)$ satisfying \erf{Steer}, then there is a model with $\wp(b|\hat{B},\xi)$ that satisfies \erf{Bell-local1}; simply make $\wp(b|\hat{B},\xi)=P(b|\hat{B};\rho_\xi)$ for all $\hat{B},\xi$. Since no steering implies no Bell nonlocality, we see that if a state is Bell nonlocal, then it implies that it is also steerable.
Hence Bell nonlocality is a stronger concept than steerability.
% That is, the Bell nonlocal states are a subset of the steerable states.

Similarly, the condition for separability implies the condition for no steering. If there is a model with $P(a|\hat{A};\sigma_\xi)$ satisfying \erf{separable1}, then there is a model with $\wp(a|\hat{A},\xi)$ that satisfies \erf{Steer}; simply make $\wp(a|\hat{A},\xi)=P(a|\hat{A};\sigma_\xi)$ for all $\hat{A},\xi$. Thus, steerability is also a stronger concept than nonseparability.  At least one
of these relations must be ``strictly stronger than'', because
Bell-nonlocality is strictly stronger than nonseparability
\cite{WerPRA89}.
In the following sections we prove
that in fact steerability is strictly stronger than
nonseparability, {\em and} strictly
weaker than Bell-nonlocality.

%Thus we also see that the steerable states are a subset of the entangled states.

%It is possible to derive the condition for no steering directly from the condition for Bell-locality.  This can be seen as follows.  If \erf{Bell-local1} is satisfied then it is straightforward to show that the RHS of \erf{Steer} will also be satisfied.  Simply make $\wp(b|\hat{B},\xi)=P(b|\hat{B};\rho_\xi)$ for all $\hat{B},\xi$. Hence, Bell-\emph{non}locality is a stronger concept than steerability. Similarly, if \erf{Steer} is satisfied, the RHS of \erf{separable1} will be satisfied simply by making $\wp(a|\hat{A},\xi)=P(a|\hat{A};\sigma_\xi)$ for all $\hat{A},\xi$. So we see that steerability is a stronger condition than nonseparability.

%Clearly %(from both our mathematical and operational definitions)
%steerability is stronger than nonseparability; steerable states must be nonseparable because if \erf{Steer} is not satisfied then there is no way that \erf{separable1} can be, since $\wp(a|\hat{A},\xi)$ is a more general distribution than $P(a|\hat{A},\sigma_\xi)$. Similarly, Bell-nonlocality is stronger than steerability; if \erf{Bell-local1} is not satisfied then there is no way that \erf{Steer} can be, since $\wp(b|\hat{B},\xi)$ is more general than $P(b|\hat{B},\rho_\xi)$.  

\section{Conditions for Steerability}
Below we derive conditions for steerability for four families of states $W$. In each example we parameterize the family of states in terms of  a mixing parameter $\eta\in \mathbb{R}$, and a second parameter that may be discrete.  In each case, the upper bound for $W$ to be a
state is $\eta  = 1$, and $W$ is a product state if $\eta = 0$,
and  (except in the last case) $W$ is linear in  $\eta$. %In each case $W$ is linear in $\eta$, the upper bound for $W$ to be a state is $\eta=1$, and $W$ is a product state if $\eta=0$. 
For the first two examples (Werner and isotropic states) the conditions derived are both necessary and sufficient for steerability.  For the other examples (inept states and Gaussian states) the conditions derived are merely sufficient for steerability.

In terms of the parameter $\eta$ we can define boundaries between different classes of states.  For example, we will make use of $\eta_{\rm Bell}$, defined by  $W^\eta$ being Bell-nonlocal \emph{iff} $\eta>\eta_{\rm Bell}$.  Similarly a state $W^\eta$ is entangled \emph{iff} $\eta>\eta_{\rm ent}$. Our goal is then to determine (or at least bound) the steerability boundaries for the above classes of states, defined by $W^\eta$ being steerable \emph{iff} $\eta>\eta_{\rm steer}$.

Crucial to the derivations of the conditions for steerability of these states is the concept of an \emph{optimal ensemble} $F^\star = \{\rho^\star_\xi \wp^\star_\xi\}$; that is, an ensemble such that if it cannot satisfy \erf{steering2} then no ensemble can satisfy it. In finding an optimal ensemble $F^\star$ we use the symmetries of $W$ and ${\mathfrak M}_\alpha$:
\begin{lemma} \label{lem_covariant}
Consider a  group $G$ with a unitary representation $\hat{U}_{\alpha\beta}(g) = \hat{U}_\alpha(g)\otimes \hat{U}_\beta(g)$ on the Hilbert space for Alice and Bob. Say that $\forall \hat{A}\in {\mathfrak M}_\alpha,\ \forall a\in \lambda(\hat{A}),\ \forall g\in G,$  we have
$\hat{U}_\alpha\dg(g)\hat{A}\hat{U}_\alpha(g) \in {\mathfrak M}_\alpha$ and
\beq
\tilde{\rho}^{\hat{U}_\alpha\dg(g)\hat{A}\hat{U}_\alpha(g)}_a = \hat{U}_\beta(g)\tilde{\rho}^{A}_a\hat{U}_\beta\dg(g).\label{lem1}
\eeq
Then there exists a $G$-covariant optimal ensemble:
$\forall g\in G,\ \{\rho^\star_\varsigma \wp^\star_\varsigma\} = \{\hat{U}_\beta(g)\rho^\star_\varsigma \hat{U}_\beta\dg(g) \wp^\star_\varsigma\}$.
\end{lemma}
\begin{proof} For specificity, consider a discrete group with order $|G|$. 
Say there exists an ensemble $F=\{\rho_\xi\wp_\xi\}$ satisfying \erf{steering2} for some map $\wp(a|\hat{A},\xi)$. Then under the conditions of Lemma \ref{lem_covariant}, $\tilde\rho_a^A$ can be rewritten as
\beq
 |G|^{-1}\sum_{g\in G} \sum_\xi \hat{U}_\beta(g)\rho_\xi \hat{U}_\beta\dg(g)  \, \wp(a|\hat{U}_\alpha\dg(g)\hat{A}\hat{U}_\alpha(g),\xi) \wp_\xi \nn
\eeq
Thus
the $G$-covariant ensemble $F^\star = \{\rho_{(g,\xi)}^\star \wp(g,\xi) \}$, with $\rho_{(g,\xi)}^\star = \hat{U}_\beta(g) \rho_\xi  \hat{U}_\beta\dg(g)$ and $\wp(g,\xi)=\wp_\xi /|G|$, satisfies \erf{steering2} with the choice
\beq
\wp^\star(a|\hat{A},(g,\xi)) =  \wp(a|\hat{U}_\alpha\dg(g)\hat{A}\hat{U}_\alpha(g),\xi).
\eeq
The analogous formulas for the case of continuous groups are elementary.
\end{proof}

Once we have determined the optimal ensemble for a given class of states (and a given measurement strategy) it remains to determine if there exists a stochastic map $\wp(a|\hat{A},\xi)$ such that \erf{steering2} is true. In each steering experiment we assume that Alice really does send Bob an entangled state.  To determine if the state is steerable, we take the perspective of a skeptical Bob and imagine that in each case Alice is attempting to cheat; that is, that she sends Bob a random state from the optimal ensemble $F^\star$ and does not perform her measurements.  She simply announces her alleged measurement results based on $\wp(a|\hat{A},\xi)$ which defines her cheating strategy. We compare the states that Bob would obtain if Alice really did send half of an entangled state and perform a measurement with those that could be prepared using an optimal ensemble and cheating strategy. %This can be checked experimentally by Bob making measurements, and looking at the correlations between his results and those reported by Alice. 

There are two possible reasons why Bob could find that his measurements results are consistent with results reported by Alice.  First, Alice could really be sending Bob half of an entangled state and steering his system via her measurements. Or, as the skeptical Bob believes, Alice could really just be sending him different pure states in each run and announcing her results based on her knowledge of this state.

%If Alice sends Bob part of a nonseparable state but it is possible that she was using a cheating strategy instead, Bob will not believe that the state is entangled (this is an example of an entangled, but not steerable, state).  Thus we need to determine the optimal ensemble and cheating strategy for Alice to use in a steering experiment and compare this with the predicted results if Alice really were sending an entangled state. 

Now if the optimal ensemble (which we are assuming Bob is clever enough to determine) can explain the correlations between Alice's announced results and Bob's results then the state sent by Alice is \emph{not} steerable.  However, if the best cheating strategy that Alice could possibly use is insufficient to explain the correlations then Bob must admit that Alice has sent him part of an entangled state.  Furthermore, if he makes this admission, the state must be steerable. 
 
%We expect that Bob will find that his state will agree with the results reported by Alice. However, over many runs of the experiment, Bob can compare his results with Alice's optimal cheating strategy (which we assume Bob is clever enough to figure out), if he finds that his local state is inconsistent with Alice's optimal cheating strategy then he will know that she must really be sending him part of an entangled state (or else the results she reported would not agree with Bob's measurements).

%Once we have determined the optimal ensemble for a given class of states (and a given measurement strategy) it remains to compare its predictions with those of QM; if the optimal ensemble cannot replicate the results of QM then no ensemble can, and the states are steerable.  %At the steering boundary the state $\tilde{\rho}_a^A$ predicted using the optimal ensemble will be equal to $\tilde{\rho}_a^Q=\bra{a}W\ket{a}$ predicted by QM. Below the steering boundary, $\tilde{\rho}_a^A$ will be more mixed than $\tilde{\rho}_a^Q$. Thus to find the steering boundary we compare the quantity $\bra{a}\tilde{\rho}_a\ket{a}$ evaluated using $\tilde{\rho}_a=\tilde{\rho}_a^A$ and $\tilde{\rho}_a=\tilde{\rho}_a^Q$. 

\subsection{Werner states}

This family of states in $\mathbb{C}_d\otimes \mathbb{C}_d$ was introduced by Werner in Ref. \cite{WerPRA89}.
As mentioned above, we parametrize it by $\eta\in\mathbb{R}$ such that $W^\eta_d$ is linear in $\eta$,  it is a product state for $\eta=0$, and is a state at all only for $\eta \leq 1$:
 \beq
W_d^\eta=\left(\frac{d-1+\eta}{d-1}\right)\frac{\mathbf{I}}{d^2}-\left(\frac{\eta}{d-1}\right)\frac{\mathbf{V}}{d}.\label{WernerStates}\eeq
Here $\mathbf{I}$ is the identity and $\mathbf{V}$ the
``flip" operator defined by $\mathbf{V}\ket{\varphi}\otimes\ket{\psi}\equiv \ket{\psi}\otimes\ket{\varphi}$.
Defining $\Phi=(1-(d+1)\eta)/d$ allows one to reproduce Werner's notation
\cite{WerPRA89} for these states.
Werner states are nonseparable \emph{iff} $\eta > \eta_{\rm ent} = 1/(d+1)$ \cite{WerPRA89}.
For $d=2$, the Werner states violate the Clauser-Horne-Shimony-Holt (CHSH) inequality \emph{iff} $\eta > 1/\sqrt{2}$  \cite{HorEtalPLA95}. This places an {\em upper} bound on $\eta_{\rm Bell}$. % defined by  $W^\eta_d$ being Bell-nonlocal iff $\eta>\eta_{\rm Bell}$. 
For $d>2$ only the trivial upper bound\footnote{This is because no Bell inequality has been found that the Werner states violate for $d>2$. It is only an upper bound because this is not a test of all possible Bell inequalities.} of $1$ is known. However, Werner found a lower bound on $\eta_{\rm Bell}$ of $1-1/d$ \cite{WerPRA89}, which is strictly greater than $\eta_{\rm ent}$.

Now let us consider the possibility of steering Werner states.  
We allow Alice all possible measurement strategies: ${\mathfrak M}_\alpha = {\mathfrak D}_\alpha$, and without loss of generality take the projectors to be rank-one: $\hat{\Pi}_a^A = \ket{a}\bra{a}$. 
For Werner states, the conditions of Lemma \ref{lem_covariant} are then satisfied for the $d$-dimensional unitary group ${\mathfrak U}(d)$. % of $d$-dimensional unitary transformations.
Specifically,  $g\to \hat{U}$, and  $\hat{U}_{\alpha\beta}(g) \to \hat{U}\otimes \hat{U}$ \cite{WerPRA89}. Again without loss of generality we can take the optimal ensemble to consist of pure states, in which case there is a unique covariant optimal ensemble, $F^\star =
\cu{\ket{\psi}\bra{\psi}d\mu_{\rm Haar}({\psi})}$, where $d\mu_{\rm Haar}(\psi)$ is the Haar measure over ${\mathfrak U}(d)$. 

If Alice were to make any projective measurement of her half of a Werner state and obtain the result $a$, Bob's unnormalized conditioned state would be given by
\bqa \tilde{\rho}_a^A&=&{\rm Tr}_A\left[(\hat{\Pi}_a^A\otimes \mathbf{I})W_d^\eta\right]=\bra{a}W_d^\eta\ket{a}\nonumber\\
&=&\left(\frac{d-1+\eta}{d\left(d-1\right)}\right)\frac{\mathbf{I}}{d}-\left(\frac{\eta}{d\left(d-1\right)}\right)\ket{a}\bra{a}.
\label{WernerBobconditioned}
\eqa This is a state proportional to the completely mixed state \emph{minus} a term proportional to the state Alice's system is projected into by her measurement.

We now determine if it is possible for Alice to simulate this conditioned state using the optimal ensemble $F^\star$ and an optimal cheating strategy defined by $\wp^\star(a|\hat{A},\psi)$. That is, we imagine that in each run of the experiment Alice simply sends Bob a state $\rho_\psi=\ket{\psi}\bra{\psi}$ drawn at random from $F^\star =
\cu{\ket{\psi}\bra{\psi}d\mu_{\rm Haar}({\psi})}$.  When asked to perform a measurement $\hat{A}$ and announce her result, she uses $\wp^\star(a|\hat{A},\psi)$ (which is based on her knowledge of $\rho_\psi=\ket{\psi}\bra{\psi}$) to determine her answer. In testing whether this is actually what Alice could be doing, we only need to consider the quantity 
\beq \label{wernersteer}
\bra{a}\tilde{\rho}_a^A\ket{a}= \frac{(1-\eta)}{d^2}.
\eeq
This is due to the form of $\tilde{\rho}_a^A$ noted above in \erf{WernerBobconditioned}.% it is a mixture made up of the completely mixed state \emph{minus} $\ket{a}\bra{a}$.%the eigenstate of $\hat{A}$ with the maximum overlap with $\ket{\psi}\bra{\psi}$.

If (on average) the strategy used by Alice with the ensemble $F^\star$ produces the correct overlap with the state $\ket{a}\bra{a}$ then \erf{steering2} will hold and steering is not possible. Thus Alice makes use of the overlap with $\ket{a}\bra{a}$ of the random states $\ket{\psi}$ in determining the optimal $\wp(a|\hat{A},\psi)$. 

Since Alice's goal is to simulate $\tilde{\rho}_a^A$, as defined in \erf{WernerBobconditioned}, she will determine which of the eigenstates of $\hat{A}$ has the \emph{least} overlap with $\ket{\psi}\bra{\psi}$ in each run of the experiment and announce the eigenvalue associated with that eigenstate as her result. On average Bob would then find that his conditioned state has the least possible overlap with $\ket{a}\bra{a}$.  %This is a sensible strategy for Alice to employ because the Werner states are anticorrelated and thus Bob's state would be projected away from the state $\ket{a}\bra{a}$ in each run (that is, have the smallest overlap with $\ket{a}\bra{a}$).  %toward the eigenstate that is furthest from the eigenstate that Alice's system would be projected into if she really were to perform a measurement. 
Writing this explicitly, the optimal distribution is given by
\beq \wp^\star(a|\hat{A},\psi)=\Bigg\{
\begin{array}{l}
  1\ {\rm if}\ \langle \psi| \hat{\Pi}_a^A |\psi\rangle < \langle \psi| \hat{\Pi}_{a'}^A |\psi\rangle\  \forall\ a'\neq a\\ 
  0\ {\rm otherwise}.
\end{array}\label{WernerMethod}%
\eeq 
It is straightforward to see that this ensemble is normalized, that is, $\forall \hat{A},\psi$,
\beq
\sum_a \wp^\star(a|\hat{A},\psi) = 1.
\eeq

%is the optimal strategy for simulating $\tilde{\rho}_a^A$ because the form of $\tilde{\rho}_a^A$ is proportional to the identity \emph{minus} a term proportional to the eigenstate of $\hat{A}$ with the maximum overlap with $\ket{\psi}\bra{\psi}$.
Clearly the optimal distribution, $\wp^\star(a|A,\psi)$, is the distribution %e optimal $\wp(a|A,\psi)$ is the one 
that will predict the same overlap with $\ket{a}\bra{a}$ as that given by \erf{wernersteer}. This occurs at precisely the steering boundary $\eta_{\rm steer}$. %, defined by $W^\eta_d$ being steerable iff $\eta>\eta_{\rm steer}$.  
When $\eta < \eta_{\rm steer}$ steering cannot be demonstrated, as it is possible that Alice is using a cheating strategy to simulate Bob's conditioned state.  This means that Alice's optimal cheating strategy could actually make Bob believe that his conditioned state has a smaller overlap with $\ket{a}\bra{a}$ than would be expected from \erf{wernersteer}.  %\footnote{This is not what Alice wishes to do however, her goal is to make Bob believe that his conditioned state has the \emph{same} overlap with $\ket{a}\bra{a}$ as predicted from \erf{wernersteer}. }. %Recall that since the Werner states are anticorrelated, Alice's goal is to simulate as \emph{small} an overlap with $\ket{\psi}\bra{\psi}$ as possible. 
In this case Alice could correctly simulate $\tilde{\rho}_a^A$ simply by introducing the appropriate amount of randomness to her responses (i.e. increase the overlap to the correct size by choosing a different $\wp(a|\hat{A},\psi)$).% non-optimal).  
To reiterate, when $\eta < \eta_{\rm steer}$ it is possible that Alice is performing a classical strategy which is consistent with Bob's results, so he will not believe that the state is genuinely steerable.

%In order to specify the optimal $\wp(a|\hat{A},\psi)$ we consider the form of $\tilde{\rho}_a^A$, the optimal ensemble and an arbitrary measurement $\hat{A}$. %=\sum_a a\hat{\Pi}_a^A$ (since we are allowing $\hat{A} \in \mathfrak{D}_\alpha$, and the optimal ensemble is uniformly distributed, this is sufficient to determine the steering boundary in general). 
% Alice's goal is to simulate $\tilde{\rho}_a^A$ as defined in \erf{WernerBobconditioned}. Therefore, she will determine which of the eigenstates of $\hat{A}$ that is \emph{furthest} from $\ket{\psi}\bra{\psi}$ in each run of the experiment and announce the eigenvalue associated with that eigenstate as her result. That is, 
%\beq \wp^\star(a|\hat{A},\psi)=\Bigg\{
%\begin{array}{l}
%  1\ {\rm if}\ \langle \psi| \hat{\Pi}_a^A |\psi\rangle < \langle \psi| \hat{\Pi}_{a'}^A |\psi\rangle\ \forall\ a\neq a'\\
%  0\ {\rm otherwise}.
%\end{array}\label{WernerMethod}%
%\eeq 
%This is the optimal strategy for simulating $\tilde{\rho}_a^A$ because the form of $\tilde{\rho}_a^A$ is proportional to the identity \emph{minus} a term proportional to the eigenstate of $\hat{A}$ with the maximum overlap with $\ket{\psi}\bra{\psi}$.

To find the form of $\eta_{\rm steer}$ we compare with Werner's result \cite{WerPRA89} for the lower bound on $\eta_{\rm ent}$. We find that  he actually used the construction outlined above.  His LHVs for Bob's system were in fact the LHSs used in the optimal ensemble $F^\star$. Werner shows that for any positive normalized distribution $\wp(a|\hat{A},\psi)$,
\beq \label{WernerGeometry}
\bra{a}\int d\mu_{\rm Haar}(\psi) \ket{\psi}\bra{\psi} \wp(a|\hat{A},\psi) \ket{a} \geq 1/d^3.
\eeq
The equality is attained for the optimal $\wp^\star(a|\hat{A},\psi)$ specified by \erf{WernerMethod} (this produces the smallest possible predicted overlap with $\ket{a}\bra{a}$).

Now to determine when \erf{steering2} is satisfied by $F^\star$ (and thus to determine $\eta_{\rm steer}$) we simply compare \erf{WernerGeometry} with \erf{wernersteer}.  We find that Alice cannot simulate the correct overlap with $\ket{a}\bra{a}$ \emph{iff}
\beq (1-\eta)/d^2 < 1/ d^3.\eeq %This is the condition for steerability, as when this condition is satisfied it is not possible for Alice to use a classical strategy to simulate the correlations with Bob's results. When this is the case, the optimal cheating strategy always predicts a greater overlap with $\ket{a}\bra{a}$ than that predicted using \erf{wernersteer}; there is no valid strategy for Alice to simulate a smaller (correct) overlap with $\ket{a}\bra{a}$. Hence Bob will accept that the state $W_d^\eta$ is genuinely steerable. %In this case we are above the steering boundary and Alice cannot introduce randomness to produce the correct overlap with $\ket{a}\bra{a}$ because there is no valid $\wp(a|\hat{A},\psi)$ to allow her to do so (by definition, the optimal $\wp(a|\hat{A},\psi)$ is the best she can do, so when this does not work, steering is possible).
 Hence we see that for Werner states \beq \eta_{\rm steer}= 1-\frac{1}{d}.\label{WernerSteerCond}\eeq

Recently a new lower bound for $\eta_{\rm Bell}$ was found for $d=2$ by Ac\'in \emph{et al.} \cite{AciEtalPRA06}, greater than $\eta_{\rm steer}$, as shown in
Fig.~\ref{combinedfigures}. Ref.~\cite{AciEtalPRA06} makes use of a connection with Grothendieck's constant (a mathematical constant from Banach space theory) to develop a local hidden variable model for projective measurements when $d=2$. Ac\`{i}n \emph{et al.} show that for two-qubit Werner  states %(and thus also for two qubit isotropic states)
 \beq 0.7071\approx 1/\sqrt{2}\geq\eta_{\rm Bell}\geq 1/K_g(3)\approx 0.6595,\eeq where $K_g(3)\approx 1.5163$ is Grothendieck's constant of order 3. Bounds on $K_g(3)$ ensure that for $d=2$ Werner states $\eta_{\rm Bell}\geq 0.6595$.  Using \erf{WernerSteerCond}, we see that when $d=2$, $\eta_{\rm steer}=1/2$. This proves that steerability is
strictly weaker than Bell-nonlocality as $\eta_{\rm steer} < 0.6595 \leq \eta_{\rm Bell}$. %as it is not possible for the two boundaries $\eta_{\rm Bell}$ and $\eta_{\rm steer}$ to overlap.  
 It is also well known that for $d=2$, $\eta_{\rm ent} =1/3$ which is strictly less than $\eta_{\rm steer}$.  Thus using the $d=2$ Werner states as an example we also see that steerability is strictly
stronger than non-separability. This clear distinction between the three classes can be seen on the left hand axis of Fig. 2 (a).

\begin{figure}
\begin{center}
\includegraphics[width=8.0cm]{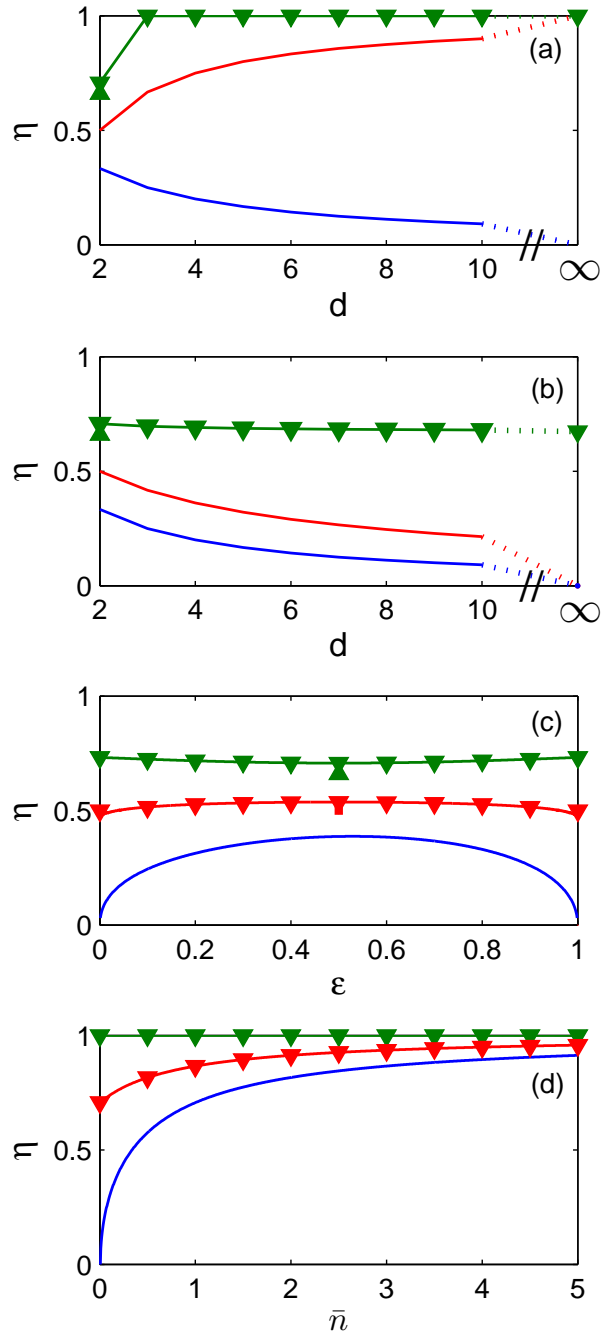}
\end{center}
\caption{(Color on-line.) Boundaries between classes of entangled
states for Werner (a) and isotropic (b) states $W^\eta_d$, inept states $W_\epsilon^\eta$ (c)  and two-mode symmetric Gaussian states $W_{\bar n}^\eta$ (d). 
The bottom (blue) line is $\eta_{\rm ent}$, above which states are
entangled. The next (red) line is $\eta_{\rm steer}$, above which
states are steerable. In cases (c) and (d) the down-arrows indicate that we have only an upper bound on $\eta_{\rm steer}$. The top (green) line with down-arrows is an upper
bound on $\eta_{\rm Bell}$, above which states are
Bell-nonlocal. The up-arrows in cases (a) and (b) are lower bounds on $\eta_{\rm Bell}$
for $d=2$. This lower bound establishes that the classes are strictly distinct. In cases (a) and (b), dots join values at
finite $d$ with those at $d=\infty$. The separate point in (c) is explained at the end of Sec. \ref{notes:PES}.}%, and the seemingly incongruous points for small $\bar{n}$ in (d) are explained at the end of Sec. \ref{GaussSym}.}
\label{combinedfigures}\end{figure}

\subsection{Isotropic states\label{notes:Iso}}

The isotropic states, which were introduced in \cite{HorHorPRA99}, can be parametrized identically to the Werner states; that is, in terms of their dimension $d$ and a mixing parameter $\eta$:
\beq W_d^\eta=(1-\eta){\mathbf{I}}/{d^2}+\eta \mathbf{P}_+.
\label{isotropicstates}\eeq Here
$\mathbf{P}_+=\ket{\psi_+}\bra{\psi_+}$, where $\ket{\psi_+}
=\sum_{i=1}^{d}\ket{i}\ket{i}/\sqrt{d}$ is a maximally entangled state. In fact, for $d=2$ it is straightforward to verify that the isotropic
states are identical to Werner states up to local unitaries. Isotropic states are
nonseparable \emph{iff}  $\eta > \eta_{\rm ent} = 1/(d+1)$
\cite{HorHorPRA99}. 

A non-trivial upper bound on $\eta_{\rm Bell}$
for all $d$ is known; in Ref. \cite{ColEtalPRL02} it is shown that a Bell inequality is certainly violated by
a $d$-dimensional isotropic state if \beq \eta >\frac{2}{I_d({\rm QM})}\geq\eta_{\rm Bell},\label{BVIso}\eeq where 
$I_d({\rm QM})$ is defined as \beq
I_d({\rm QM})=4d\sum_{k=0}^{[d/2]-1}\left(1-\frac{2k}{d-1}\right)\left(q_k-q_{-(k+1)}\right),\eeq
and $q_k=1/\{2d^3\sin^2[\pi(k+1/4)/d]\}.$ Collins \emph{et al}.
\cite{ColEtalPRL02} go on to show that in the limit as
$d\rightarrow\infty$ the limiting value this upper bound on $\eta_{\rm Bell}$ %of $\eta$ for a Bell
%inequality violation to occur 
approaches $\pi^2/(16 \times {\rm
Catalan})\approx 0.6734$, where ${\rm Catalan}\approx0.9159$ is Catalan's
constant.

In determining steerability we again allow Alice all possible measurement strategies: ${\mathfrak M}_\alpha = {\mathfrak D}_\alpha$, and take the projectors to be rank-one: $\hat{\Pi}_a^A = \ket{a}\bra{a}$. 
The isotropic states have the symmetry property that they are invariant under transformations of the form $\hat{U}^*\otimes \hat{U}$, hence the conditions of Lemma \ref{lem_covariant} are again satisfied for the $d$-dimensional unitary group ${\mathfrak U}(d)$. % of $d$-dimensional unitary transformations.
In this case,  $g\to \hat{U}$, and  $\hat{U}_{\alpha\beta}(g) \to \hat{U}^*\otimes \hat{U}$. Thus we can again take the optimal ensemble to be %consist of pure states, in which case there is a unique covariant optimal ensemble, 
$F^\star =
\cu{\ket{\psi}\bra{\psi}d\mu_{\rm Haar}({\psi})}$.% where $d\mu_{\rm Haar}(\psi)$ is the Haar measure over ${\mathfrak U}(d)$. We find that due to the similarities between the isotropic and Werner states, the optimal ensemble $F^\star$ is the same for both these classes of states.

%To determine steerability we must construct an optimal ensemble for the class of states of interest and the measurement strategy to be used.  Based on the symmetry of the isotropic states we take an ensemble consisting of pure states, $F^\star =
%\cu{\ket{\psi}\bra{\psi}d\mu_G({\psi,m})}$, where $d\mu_G(\psi,m)$ is a Gaussian measure over ${\mathfrak U}(d)$ (see appendix \ref{GaussApp} for details). The conditions of Lemma 1 are then satisfied for the $d$-dimensional unitary group ${\mathfrak U}(d)$ with $g\to \hat{U}$, and  $\hat{U}_{\alpha\beta}(g) \to \hat{U}^*\otimes \hat{U}$; the ensemble $F^\star$ is indeed optimal. 

Now consider the conditioned state that Bob would obtain if Alice were to make a measurement $\hat{A}$ on her half of $W_d^\eta$,
\bqa \tilde{\rho}_a^A&=&{\rm Tr}_A\left[(\hat{\Pi}_a^A\otimes \mathbf{I})W_d^\eta\right]\nonumber\\
&=&\left(\frac{1-\eta}{d}\right)\frac{\mathbf{I}}{d}+\frac{\eta}{d}\ket{a}\bra{a}.
\label{IsoBobconditioned}
\eqa This is a state proportional to the completely mixed state \emph{plus} a term proportional to the state Alice's system would be projected into by her measurement. Note the similarity with the Werner state example, where the conditioned state was proportional to the completely mixed state \emph{minus} a term proportional to $\ket{a}\bra{a}$.  This difference arises because the isotropic states are symmetrically correlated rather than anti-symmetrically correlated as in the Werner state example.

Again we wish to determine if it is possible for Alice to simulate the conditioned state $\tilde{\rho}_a^A$ using the optimal ensemble $F^\star$ and a cheating strategy defined by an optimal distribution $\wp^\star(a|\hat{A},\psi)$. %In this case we will first look at some properties of the optimal ensemble. 

Imagine that in each run of a steering experiment Alice simply sends Bob a  state $\ket{\psi}\bra{\psi}$ drawn at random from $F^\star =
\cu{\ket{\psi}\bra{\psi}d\mu_{G}({\psi,m})}$.  When asked to perform a measurement $\hat{A}$ and announce her result, she uses $\wp^\star(a|\hat{A},\psi)$ to determine her answer. In testing whether this is actually what Alice could be doing, we again only need to consider the quantity 
\beq \bra{a}\tilde{\rho}_a^A\ket{a}= \frac{\eta}{d}+\frac{(1-\eta)}{d^2}\label{Isosteer}
.\eeq 
In this case Alice's strategy is similar to the Werner state example, except now she wants to to simulate the \emph{maximum} possible overlap with $\ket{a}\bra{a}$ (due to the form of $\tilde{\rho}_a^A$). %is because the conditioned state $\tilde{\rho}_a^A$ is a mixture of the completely mixed state \emph{plus} a term proportional to $\ket{a}\bra{a}$. 
Therefore, Bob will only concede that $W_d^\eta$ is steerable if the maximum overlap with $\ket{a}\bra{a}$ predicted using the ensemble $F^\star$ and the optimal cheating strategy $\wp^\star(a|A,\psi)$ is less than that predicted by \erf{Isosteer}. In this case there would be no possible classical strategy that Alice could possibly be using to simulate the correlations with Bob's results. %This is because the conditioned state $\tilde{\rho}_a^A$ is proportional to the completely mixed state \emph{plus} a term proportional to Alice's alleged result. 
Identical predictions for the overlap with $\ket{a}\bra{a}$ will again occur precisely at the steering boundary $\eta_{\rm steer}$, which occurs when $\wp^\star(a|\hat{A},\psi)$ is used.

The optimal $\wp(a|\hat{A},\psi)$ is defined in a similar manner to the Werner state example. However, in each run of the experiment % Werner state example, in order to specify the optimal $\wp(a|\hat{A},\psi)$ we must consider the form of $\tilde{\rho}_a^A$, the optimal ensemble and an arbitrary measurement $\hat{A}=\sum_a a\hat{\Pi}_a^A$ (again since we are allowing $\hat{A} \in \mathfrak{D}_\alpha$ and the optimal ensemble is uniformly distributed this is sufficient to determine the steering boundary in general).  Alice's goal is to simulate $\tilde{\rho}_a^A$ as defined in \erf{IsoBobconditioned}. Therefore, she will 
Alice now determines which of the eigenstates of $\hat{A}$ is \emph{closest} to $\ket{\psi}\bra{\psi}$ and announces the eigenvalue associated with that eigenstate as her result. That is, 
\beq \wp^\star(a|\hat{A},\psi)=\Bigg\{
\begin{array}{l}
  1\ {\rm if}\ \langle \psi| \hat{\Pi}_a^A |\psi\rangle > \langle \psi| \hat{\Pi}_{a'}^A |\psi\rangle\  \forall\ a'\neq a\\ 
  0\ {\rm otherwise}.
\end{array}\label{IsoMethod}%
\eeq

%This is the optimal strategy for Isotropic states because $\tilde{\rho}_a^A$ is proportional to the identity \emph{plus} a term proportional to the eigenstate of $\hat{A}$ with the maximum overlap with $\ket{\psi}\bra{\psi}$.  Note that this is precisely opposite to the strategy used in the Werner state example.  This is due to the particular symmetry of each of the families of states; the Werner states are anticorrelated states so that the conditioned state $\tilde{\rho}_a^A$ is a mixture of the completely mixed state \emph{minus} the state $\ket{a}\bra{a}$, whereas the isotropic states are correlated so that  $\tilde{\rho}_a^A$ is a mixture of the completely mixed state \emph{plus} the state $\ket{a}\bra{a}$.  This is where the difference in the optimal strategies arises.
 
To test if \erf{steering2} holds, Alice and Bob would need to run the experiment many times and compare $\bra{a}\tilde{\rho}_a^A\ket{a}$ with the quantity
\beq \bra{a}\int d\mu_{\rm Haar} \ket{\psi}\bra{\psi}\wp^\star(a|A,\psi)\ket{a}\label{p(aA)}.\eeq %Recall that if this is less than \erf{Isosteer}, there is no classical strategy that Alice could be using in the experiment; \erf{steering2} is violated and steering is possible for isotropic states. \erf{p(aA)} 
 This can be written as
\beq \bra{a}\int_a d\mu_{\rm Haar}(\psi)\ket{\psi}\bra{\psi}a\rangle =\int_a d\mu_{\rm Haar}(\psi)|\bra{a}\psi\rangle|^2,\label{IsoavA}
\eeq where the subscript $a$ on the integral means that in the integral only those states with $\left|\bra{a}\psi\rangle\right|$ greater than all others will contribute. As shown in Appendix \ref{App:optens}, a random state $\ket{\psi}$ from the ensemble $F^\star$ can be described by the unnormalized state
\beq \ket{\tilde{\psi}}=m\ket{\psi}=\frac{1}{\sqrt{d}}\sum_{j=1}^d z_j\ket{\phi_j},\label{isoensemblestate}
\eeq where the $z_j$ are mutually independent complex Gaussian random variables with zero mean and zero second moments except for $\langle z_j^\ast z_k\rangle=\delta_{jk}$. That is, we can replace the Haar measure $d\mu_{\rm Haar}(\psi)$ by $d\mu_G(\psi,m)=d\mu_G(m)d\mu_{\rm Haar}(\psi)$. In terms of the variables $\{z_j\}$,this can be expressed as
\beq d\mu_G(\psi,m)\rightarrow \pi^{-d}\exp\left(-\sum_{i=1}^d z_i^2\right)d^2z_1...d^2z_d.
\eeq 

Now using the Gaussian measure $d\mu_G(\psi,m)$ to describe the ensemble $F^\star$, we can rewrite \erf{IsoavA} as
\bqa \int_a d\mu_{\rm Haar}(\psi)|\bra{a}\psi\rangle|^2 &\!\!=&\!\!\!\int_a \!d\mu_{\rm Haar}(\psi)|\bra{a}\psi\rangle|^2\frac{\int d\mu_G(m)m^2}{\int d\mu_G(m)m^2}\nonumber\\
&\!\!=&\!\frac{\int_ad\mu_G(\psi,m)|\bra{a}\tilde{\psi}\rangle|^2}{\int d\mu_G(m)m^2}.\label{IsoAv}
\eqa
It is straightforward to show that the denominator equals one (see Appendix \ref{App:norm}), and hence we can evaluate the numerator (left to  Appendix \ref{App:Int}) to find that
\beq\int_ad\mu_G(\psi,m)|\bra{a}\tilde{\psi}\rangle|^2=\frac{H_d}{d^2},\label{H_d}
\eeq where $H_d=1+1/2+1/3+\ldots+1/d$ is the Harmonic series.

Thus we find that for any positive normalized distribution $\wp(a|\hat{A},\psi)$ we must have
\beq
\bra{a}\int d\mu_{\rm Haar}(\psi)\ket{\psi}\bra{\psi}\wp(a|\hat{A},\psi)\ket{a} \leq \frac{H_d}{d^2},\label{IsoCondition}
\eeq with the equality obtained for the optimal $\wp^\star(a|\hat{A},\psi)$ as defined in \erf{IsoMethod}.
%To determine when \erf{steering2} is satisfied by $F^\star$ we recall \erf{IsoBobconditioned} for which it is simple to show that, for
%any $\hat{A} \in {\mathfrak D}_\alpha$ and $a \in \lambda(A)$,
%\beq \label{Isosteer}
%\bra{a}\tilde{\rho}_a^A\ket{a}= \frac{\eta}{d}+\frac{(1-\eta)}{d^2}.
%\eeq
Comparing this with \erf{Isosteer} we see that steering can be demonstrated \emph{iff}
\beq\frac{\eta}{d}+\frac{(1-\eta)}{d^2} > \frac{H_d}{d^2}.\eeq Thus for isotropic states \beq \eta_{\rm steer} =\frac{H_d-1}{d-1}\  \begin{array}{c} \longrightarrow \\ \mbox{\rm \small large\ d} \end{array}\ \   \frac{\ln(d)}{d}.
\label{IsoSteerCond}\eeq { For $d=2$ the isotropic states are equivalent (up to local unitaries) to the Werner states, and we again find that $\eta_{\rm steer}=1/2$ which is strictly less than $\eta_{\rm Bell}$ and strictly greater than $\eta_{\rm ent}$.}   For $d>2$, $\eta_{\rm steer}$ is greater than $\eta_{\rm ent}$ and significantly less than an upper bound on $\eta_{\rm Bell}$.  This is shown in Fig. \ref{combinedfigures}(b).  For large $d$ we see that both $\eta_{\rm steer}$ and $\eta_{\rm ent}$ tend to zero, however, $\eta_{\rm steer}$ approaches zero more slowly; it is larger than $\eta_{\rm ent}$ by a factor of $\ln(d)$  \cite{AciEtalPRL07}.

\subsection{Inept states\label{notes:PES}}

We now consider a family of states with less symmetry than the previous examples. This makes the analysis more difficult, meaning that we cannot find $\eta_{\rm steer}$ exactly. However, making use of the symmetry properties of the states allows us to find an upper bound on $\eta_{\rm steer}$.   We define a family of two-qubit states  by 
\begin{equation}W_\epsilon^\eta=\eta\ket{\psi}\bra{\psi}+(1-\eta)\rho_\alpha\otimes\rho_\beta,
\label{PESmatrix}
\end{equation}
where \beq \ket{\psi}=\sqrt{1-\epsilon}\ \ket{0_\alpha0_\beta}+\sqrt{\epsilon}\
\ket{1_\alpha1_\beta},\eeq and the reduced states $\rho_{\alpha(\beta)}$ are found by
partial tracing with respect to Bob (Alice). That is,
\beq \rho_{\alpha(\beta)}={\rm Tr}_{\beta(\alpha)}[\ket{\psi}\bra{\psi}].\eeq   As in the previous examples, this is a two-parameter family of states; the parameter $\eta$ is again a mixing parameter, and the parameter
$\epsilon$ determines how much entanglement is present in the state
$\ket{\psi}$.  Note that when $\epsilon=1/2$ these states are equivalent to the two-dimensional Werner and isotropic states.

This family of states was studied in Ref. \cite{JonEtalPRA05} in the context of distributing entanglement. The authors considered an inept company attempting to distribute pure entangled states $\ket{\psi_\epsilon}$ to many pairs of parties. However, they mixed up the addresses some fraction $1-\eta$ of the time, meaning that on average the company would actually distribute mixed entangled states of the form of \erf{PESmatrix}. Hence we will refer to this family of states as `inept' states.

As noted above, the inept states are a family of two-qubit states, which means that it is possible to evaluate $\eta_{\rm ent}$ analytically. This was done in Ref. \cite{JonEtalPRA05} leading to the following condition for
nonseparability of inept states, \beq
\eta>\eta_{\rm ent}=\frac{\epsilon(1-\epsilon)}{\epsilon(1-\epsilon)+\sqrt{\epsilon(1-\epsilon)}}.
\label{PESSep}\eeq

Ref. \cite{JonEtalPRA05} also considers Bell nonlocality of the state matrix $W_\epsilon^\eta$ by testing if a violation of the CHSH inequality
\cite{ClaEtalPRL69} occurs.  This was done using the method of \cite{HorEtalPLA95} for determining the optimal violation of the CHSH inequality for two-qubit states. One finds that the state  $W_\epsilon^\eta$  violates the CHSH inequality if and only if \beq \eta>
 \frac{4\epsilon^2-4\epsilon+1-\sqrt{4\epsilon^2-4\epsilon+3}}{4\epsilon^2-4\epsilon-1}\geq\eta_{\rm Bell} \label{PESBellviol}.\eeq 

Now, in order to demonstrate steering we must specify a measurement strategy. In the two previous examples we have used the complete set of projective measurements, $\mathfrak{M}_\alpha=\mathfrak{D}$. This would be a suitable measurement strategy to allow us to define an optimal ensemble, however, in order to make our task simpler we will consider a more restricted set of measurements.  We note that states defined by \erf{PESmatrix} have the symmetry property that they are invariant under simultaneous contrary rotations about the $z$-axes. This immediately suggests a restricted measurement scheme; we allow all measurements in the $xy$-plane but only allow a single measurement along the $z$-axis. That is, Alice's measurement scheme
is given by ${\mathfrak M}_\alpha=\cu{\s{z}} \cup \cu{\hat{\sigma}_\theta:\theta \in [0,2\pi)}$, where
\beq
\hat{\sigma}_\theta=\s{x}\cos(\theta)
+ \s{y}\sin(\theta).
\eeq In this case the conditions for Lemma 1 are satisfied for the Lie group $G$ generated by $(1/2)\hat{\sigma}_z\otimes\mathbf{I}-(1/2)\mathbf{I}\otimes\hat{\sigma}_z$ (see Appendix \ref{IneptOptimalApp}).

This is a more restricted scheme than we have considered so far, but will be sufficient to demonstrate steerability if \erf{steering2} does not hold (since it must hold for \emph{all} measurements to preclude steering). Thus we are only considering an upper bound on $\eta_{\rm steer}$ (the boundary between steerable and non steerable states using all projective measurements).

We now consider the optimal ensemble for this restricted set of measurements. We use an ensemble of pure states $F=\cu{\ket{\psi}\bra{\psi}d\mu(\psi)}$, where
\bqa \ket{\psi}\bra{\psi}&=&\frac{1}{2}\left(\mathbf{I}+\sqrt{1-z^2}\cos(\phi)\hat{\sigma}_x\right.\nonumber\\&&\phantom{\frac{1}{2}\mathbf{I}\mathbf{I}\mathbf{I}}\left.+\sqrt{1-z^2}\sin(\phi)\hat{\sigma}_y-z\hat{\sigma}_z\right)\label{ineptoptimal}\eqa and $d\mu(\psi)=(d\phi/2\pi)\wp(z)dz$.  It is straightforward to show that this ensemble is of the form of the optimal ensemble since the conditions for Lemma 1 hold (see Appendix \ref{IneptOptimalApp}). While this ensemble has the form of the optimal, it is not completely specified as $\wp(z)$ is still general.  Thus to find the optimal ensemble we need to determine the optimal probability distribution $\wp(z)$.  

First consider the reduced states that Bob would obtain if Alice really were to measure $\hat{\sigma}_z$ on her half of $W_\epsilon^\eta$. If she did so, and obtained the $+1$ result then Bob's state would be given by
\beq
\tilde{\rho}_{+1}^{\sigma_z}=\left(\epsilon\right)\frac{1}{2}\left[ \mathbf{I}-z_+\hat{\sigma}_z\right].\label{zplus}
\eeq  Similarly, for the $-1$ result, Bob would obtain
\beq
\tilde{\rho}_{-1}^{\sigma_z}=\left(1-\epsilon\right)\frac{1}{2}\left[ \mathbf{I}-z_-\hat{\sigma}_z\right].\label{zminus}
\eeq
%\beq
%\tilde{\rho}_{-1}^{\sigma_z}=\left(\epsilon\right)\frac{1}{2}\left[ \mathbf{I}+z_-\hat{\sigma}_z\right],\label{zminus}
%\eeq 
where the constants $z_+$ and $z_-$ are defined as
\bqa
z_+&=&1-2\eta-2\epsilon\left(1-\eta\right),\nonumber\\
z_-&=&1-2\epsilon\left(1-\eta\right).\label{zplusminus}
\eqa

Now we wish to determine if Alice could simulate these conditioned states using the ensemble $F$ and a suitable strategy $\wp(\pm 1|\hat{\sigma}_z,(z,\phi))$.  Due to the form of $\tilde{\rho}_{\pm 1}^{\sigma_z}$ the best strategy for Alice is to split the ensemble $F$ into two sub-ensembles, one to simulate $\tilde{\rho}_{+1}^{\sigma_z}$ and the other to simulate $\tilde{\rho}_{-1}^{\sigma_z}$.  Thus we can separate $\wp(z)$ into two positive distributions
\beq \wp(z)=\wp_+(z)+\wp_-(z).
\eeq %where $\wp_+(z)=f_+(z)^2$ and $\wp_-(z)=f_-(z)^2$.  
We imagine that Alice will attempt to simulate measuring $\hat{\sigma}_z$ by randomly generating states $\ket{\psi}\bra{\psi}$ using the distribution $\wp(z)$ and sending them to Bob.  If in a particular run of the experiment the state she sent Bob was from the sub-ensemble determined by $\wp_+(z)$ then she will announce the result $+1$.  Similarly if she sent Bob a state from $\wp_-(z)$ then she will announce $-1$. % That is,
%\beq \wp(\pm 1|\hat{\sigma}_z,(z,\phi))=\Bigg\{
%\begin{array}{l}
%  1\ {\rm if}\ z \in \wp_{\pm}(z)\\
%  0\ {\rm if}\ z \in \wp_{\mp}(z).
%\end{array}\label{IneptMethod}%
%\eeq 

Now if Alice uses this strategy, Bob will find on average that
\beq
\tilde{\rho}_{\pm 1}^{\sigma_z}=\frac{1}{2}\left[\mathbf{I}\int_{-1}^{+1}dz\wp_\pm(z)-\hat{\sigma}_z\int_{-1}^{+1}dz\wp_{\pm}(z)z \right].
\eeq% where $c_+=1-\epsilon$ and $c_-=\epsilon$. 
Comparing with \erfs{zplus}{zminus} we find that in order for the ensemble $F$ to be able to simulate Alice measuring $\hat{\sigma}_z$ we have the following constraints on $\wp(z)$, %which can be expressed in terms of $f_+(z)$ and $f_-(z)$ as
\bqa
\int_{-1}^{1}dz\wp_+(z)&=&\epsilon, \label{1}\\
\int_{-1}^{1}dz\wp_-(z)&=&1-\epsilon,\label{2}\\
\int_{-1}^{1}dz\wp_+(z)z&=&\epsilon z_+,\label{3}\\
\int_{-1}^{1}dz\wp_-(z)z&=&(1-\epsilon) z_- \label{4}.\eqa 

Now consider the conditioned states that Bob would obtain if Alice were to measure $\hat{\sigma}_\theta$:
\bqa \label{sigphi}
\tilde{\rho}_{\pm 1}^{\sigma_\theta}&=&\frac{1}{2}\left[\mathbf{I} \pm \eta\sqrt{\epsilon(1-\epsilon)}\cos(\theta)\hat{\sigma}_x \right.\\
&&\phantom{\frac{1}{2}}\left.\pm \eta\sqrt{\epsilon(1-\epsilon)}\sin(\theta)\hat{\sigma}_y 
- (1-2\epsilon)\hat{\sigma}_z \right].\nonumber
\eqa
How well could Alice simulate the above state using the ensemble $F^\star$ and a cheating strategy defined by $\wp(\pm 1|\hat{\sigma}_\theta,(z,\phi))$? We know that the ensemble $F^\star$ is symmetric under rotations about the $z$-axis.  So in this case Alice would use her knowledge of $\phi$ to determine the outcome to announce when asked to measure $\hat{\sigma}_\theta$. That is, if the state $\ket{\psi}\bra{\psi}$ that she sent Bob is closer to the positive axis defined by $\hat{\sigma}_\theta$ then she will announce the $+1$ result.  Similarly, if $\ket{\psi}\bra{\psi}$ is closer to the negative measurement axis then she announces $-1$.  This corresponds to
\beq \wp(\pm 1|\hat{\sigma}_\theta,(z,\phi))=\Bigg\{
\begin{array}{l}
  1\ {\rm if}\ \phi \in [\theta\mp\frac{\pi}{2}, \theta\pm\frac{\pi}{2})\\
  0\ {\rm if}\ \phi \in [\theta\pm\frac{\pi}{2}, \theta\mp\frac{\pi}{2}).
\end{array}\label{IneptMethod2}%
\eeq 
From the symmetry under rotations about the $z$-axis we can see that Alice will be able to do equally well using this strategy to simulate states prepared by any measurement $\hat{\sigma}_\theta$ in the $xy$-plane. Thus without loss of generality we set $\theta=0$ and consider the specific case where Alice allegedly measures $\hat{\sigma}_x$.  Under these conditions \erf{sigphi} reduces to
\beq
\tilde{\rho}_{\pm 1}^{\sigma_x}=\frac{1}{2}\left[\mathbf{I} \pm \eta\sqrt{\epsilon(1-\epsilon)}\hat{\sigma}_x - (1-2\epsilon)\hat{\sigma}_z \right]. \label{xplusminus}
\eeq

If Alice randomly sends Bob states from $F^\star$ and uses \erf{IneptMethod2} to determine her responses, Bob will find on average the state
\beq
\frac{1}{2}\left[\mathbf{I} \pm \frac{1}{\pi} \int_{-1}^{+1}dz\sqrt{1-z^2}\wp(z)\hat{\sigma}_x - \int_{-1}^{+1}dz\wp(z)z\hat{\sigma}_z \right]. \label{xaverage}
\eeq

We know that when $F$ is optimal, \erf{xaverage} will exactly simulate \erf{xplusminus}. In determining the optimal $F$ we must find the optimal $\wp(z)$, however, we are constrained in determining $\wp(z)$ by the fact that the ensemble must also simulate the states that Bob would obtain if Alice were to measure $\hat{\sigma}_z$.  These constraints are enforced by \erflist{1}{4}. 

Note that \erfs{1}{2} ensure that the $\hat{\sigma}_z$ term in \erf{xaverage} and \erf{xplusminus} will be the same. Therefore, to determine how well Alice's strategy can simulate \erf{xplusminus} we only need to consider the coefficient of the  $\hat{\sigma}_x$ term.  If the coefficient of this term predicted by \erf{xaverage} is as large as in \erf{xplusminus} then Alice's strategy simulates Bob's conditioned state perfectly.  Thus Bob would not believe that the state $W_\epsilon^\eta$ is genuinely steerable.  Hence we need to find the distribution $\wp^\star(z)$ which maximizes the $\hat{\sigma}_x$ coefficient in \erf{xaverage} to determine if steering is possible. That is, we wish to find the $\wp(z)$ that gives the maximum value of $(1/\pi) \int_{-1}^{+1}dz\sqrt{1-z^2}\wp(z)$.  This is equivalent to maximizing
\beq
\frac{1}{\pi}\int_{-1}^{+1}dz\sqrt{1-z^2}\left[ \wp_+(z)+\wp_-(z)\right],
\eeq subject to the constraints given by \erflist{1}{4}.

Writing $\wp_\pm(z)=f^2_\pm(z)$ for real functions $f_\pm(z)$ we can use Lagrange multiplier techniques to perform the optimization. We find that the optimal $\wp(z)$ has the unsurprising form
\beq
\wp^\star(z)=\epsilon \delta(z-z_+)+(1-\epsilon)\delta(z-z_-),
\eeq where the constants $z_\pm$ are defined in \erf{zplusminus} and $\delta(z-z')$ is the Dirac delta function. We see now why the choice of splitting the ensemble into two distributions was the best choice for Alice.  The optimal ensemble $F^\star$ is composed of pure states in two rings around the $z$-axis of the Bloch sphere; one in the $+z$-hemisphere defined by $z_+$, which on average may be used to simulate $\tilde{\rho}_{+1}^{\sigma_z}$, and the other in the $-z$-hemisphere defined by $z_-$ which may simulate $\tilde{\rho}_{+1}^{\sigma_z}$. (These comments apply to the case $0 \leq \epsilon \leq 1/2$.)
 
Using $\wp^\star(z)$ to evaluate \erf{xaverage} we find that
\bqa
\tilde{\rho}_{\pm 1}^{\sigma_x}&=&\frac{1}{2}\left[\mathbf{I} \pm \frac{1}{\pi}\left\{ \epsilon\sqrt{1-z_+^2}+(1-\epsilon)\sqrt{1-z_-^2}\right\}\hat{\sigma}_x\right.\nonumber\\ &&\phantom{\frac{1}{2}\mathbf{I} \pm}\Big.- (1-2\epsilon)\hat{\sigma}_z \bigg]. \label{xaverage2}
\eqa Finally comparing this with $\tilde{\rho}_{\pm 1}^{\sigma_x}$ given by \erf{xplusminus} we find that Alice's optimal cheating strategy fails to simulate measurements of $\hat{\sigma}_x$ when
\bqa\pi\eta\sqrt{\epsilon(1-\epsilon)}-(1-\epsilon)\sqrt{1-\left[1-2\epsilon(1-\eta)\right]^2}\nonumber\\
-\epsilon\sqrt{1-\left[1-2\eta-2\epsilon(1-\eta)\right]^2}>0.
\label{PESsteeringcond2}\eqa  Thus under these conditions we know that steering is possible using the measurement scheme $\mathfrak{M}_\alpha$.  Note that we have not determined $\eta_{\rm steer}$ as we have not considered all possible projective measurements.  However, we can make \erf{PESsteeringcond2} an equality to provide an equation for $\eta$ which is an \emph{upper bound} on $\eta_{\rm steer}$. This boundary is plotted in Fig. \ref{combinedfigures} (c).

For $\epsilon=1/2$ we know explicitly that \beq \eta_{\rm Bell} > \eta_{\rm steer} > \eta_{\rm ent},\eeq (since these states are equivalent to the $d=2$ Werner states, see Appendix \ref{App:IneptHalf}). This special case yields the isolated points at $\epsilon=1/2$ in Fig. \ref{combinedfigures} (c). For the remaining range of $\epsilon$ we find that our upper bound on $\eta_{\rm steer}$ is significantly lower than the upper bound on $\eta_{\rm Bell}$ and significantly higher than $\eta_{\rm ent}$.  This fact, taken with the known boundary values for $\epsilon=1/2$ gives us good reason to conjecture that the the three boundaries are strictly distinct for all $\epsilon\in [0,1]$.

\subsection{Gaussian states\label{notes:Gauss}}

Finally we investigate a general (multi-mode) bipartite Gaussian state $W$ \cite{GieCirPRA02}.  
The mode operators are defined as $\hat{q}_i=\hat{a}_i+\hat{a}_i^\dagger$ and $\hat{p}_i=-i(\hat{a}_i-\hat{a}_i^\dagger)$ for the position and momentum respectively.  Here $\hat{a}_i$ and $\hat{a}_i^\dagger$ are the annihilation and creation operators for the $i$th mode. For an $n$-mode state one may define a vector $\hat{R}=(\hat{q}_1,\hat{p}_1,...,\hat{q}_n,\hat{p}_n)$ which allows the commutation relations for the mode operators to be compactly expressed as
\beq
[R_i,R_j]=2i\Sigma_{\alpha\beta}^{ij}.\label{commutation}
\eeq
Here $\Sigma_{\alpha\beta}^{ij}$ are matrix elements of the symplectic matrix $\Sigma_{\alpha\beta}=\bigoplus_{i=1}^{n} J_i$ where
\beq J_i=\left(\begin{matrix}
0 &   1 \\
-1  & 0
\end{matrix} \right).
\eeq

A Gaussian state is defined by the mean of the vector of phase-space variables $\hat R$, as well as the covariance matrix (CM) $V_{\alpha\beta}$ for these variables.  The mean vector can be arbitrarily altered by local unitary operations and hence cannot determine the entanglement properties of $W$.  Thus for our purposes a Gaussian state is characterized by the CM. In (Alice, Bob) block form it appears as
\beq \label{cov-matrix}
{\rm CM}[W] =V_{\alpha\beta}=\left(\begin{matrix}
V_\alpha &   C \\
C\tp  & V_\beta
\end{matrix} \right).
\eeq
This represents a valid state \emph{iff} the linear matrix inequality (LMI) 
\beq 
V_{\alpha\beta}+i\Sigma_{\alpha\beta}\geq 0
\eeq is satisfied \cite{GieCirPRA02} . 

Rather than addressing steerability in general, we consider the case where Alice can only make  Gaussian measurements \cite{GieCirPRA02,GieEtalQIC03}, the set of which will be denoted by $\mathfrak{G}_\alpha$. Thus,  as for the previous section, since we are considering a restricted class of measurements, if we demonstrate steerability with this measurement scheme it will provide an \emph{upper bound} on $\eta_{\rm steer}$. 

A measurement $A\in \mathfrak{G}_\alpha$ is described by a Gaussian positive operator with a CM $T^A$ satisfying $T^A + i\Sigma_\alpha \geq 0$ \cite{GieCirPRA02,GieEtalQIC03}. When Alice makes such a measurement, Bob's conditioned state $\rho_a^A$  is Gaussian with a CM \cite{ZhoEtalORC96}
\beq \label{cndV}
{\rm CM}[\rho_a^A]=V_\beta^A =  V_\beta -  C\tp(V_\alpha+T^A)^{-1}C,
\eeq
which is actually independent of Alice's outcome $a$.

Our goal is to determine a sufficient condition for steerability of Gaussian states.  We do this by determining the necessary and sufficient condition for steerability with Gaussian measurements.  In the previous examples after specifying a measurement scheme we considered Bob's conditioned state if Alice were to perform a measurement and determined when it was possible for this state to be simulated by a cheating strategy.  In the following we are working toward the same goal.  If Alice were to perform a Gaussian measurement on half of the state $W$ and send the other part to Bob, then Bob's conditioned state would have a covariance matrix defined by \erf{cndV}; however, this is independent of Alice's result $a$. Thus we do not need to consider a strategy for Alice to announce correctly correlated results to Bob.  We simply need to determine when Alice could simulate the Bob's conditioned state by sending Bob states from a pure state ensemble (rather than actually sending part of $W$).  We will show that there exists an optimal ensemble of Gaussian states distinguished by their mean vectors (but sharing the same covariance matrix, which we will label $U$) which Alice could use for this task. If there exists a valid ensemble of Gaussian states defined by $U$ which can simulate $V_\beta^A$ then Bob will not believe that $W$ is entangled, and hence the state is not steerable.

Before moving to the presentation of our main result consider the following result from linear algebra theory relating to Schur complements of block matrices. The Schur complements of $P$ and $Q$ in a general block matrix
\beq \label{Schur}
B=\left(\begin{matrix}
P &   R \\
R\tp  & Q
\end{matrix} \right),
\eeq are defined as $\Delta_P=Q-R\tp P^{-1}R$ and $\Delta_Q=P-RQ^{-1}R\tp$ respectively.  The matrix $B$ is positive semidefinite (PSD), \emph{iff} both $P$ and its Schur complement are PSD (and likewise for $Q$ and its Schur complement).

The proof of our main theorem is based on the following inequality
\beq \label{LMIL}
V_{\alpha\beta} + {\bf 0}_\alpha\oplus i\Sigma_\beta \geq 0,
\eeq
and relies on the following facts:

\begin{lemma}\label{part1}
If \erf{LMIL} is true then there exists an ensemble defined by covariance matrix $U$ such that
\bqa
U+i\Sigma_\beta&\geq&0\label{LMI1},\\
V_\beta^A-U&\geq&0\label{LMI2},
\eqa
which implies that the state $W$ is not steerable.
\end{lemma}
\noindent\textbf{Proof}
See Appendix \ref{lemma2proof} for proof of this lemma.

\begin{lemma}\label{part2}
If the Gaussian state $W$ defined in \erf{cov-matrix} is {\em not} steerable by Alice's Gaussian measurements then there exists a Gaussian ensemble defined by covariance matrix $U$ such that \erf{LMI1} and \erf{LMI2} hold.
\end{lemma}
\noindent\textbf{Proof}
See Appendix \ref{lemma3proof}  for proof of this lemma.

\begin{lemma}\label{part3} If $\mathit{\forall\ A\in\mathfrak{G}}$ there exists $U$ such that \erfs{LMI1}{LMI2} hold, then \beq V_\alpha+T^A-C(V_\beta+i\Sigma_\beta)^{-1}C\tp \geq 0\label{lemma3},
\eeq must also hold.\end{lemma}
\noindent\textbf{Proof}
See Appendix \ref{lemma4proof}  for proof of this lemma.

We are now in a position to present our main theorem.

\begin{theorem}
The Gaussian state $W$ defined in \erf{cov-matrix} is {\em not} steerable by Alice's Gaussian measurements \emph{} \erf{LMIL} is true.
\end{theorem}

\begin{proof} 
By Lemma 2 we know that if \erf{LMIL} is true then $W$ is not steerable.

Now suppose that \erf{LMIL} does not hold, so that
\beq
V_{\alpha\beta}+0_\alpha\oplus i\Sigma_\beta \ngeq 0.
\eeq
Since we know that $V_\beta+i\Sigma_\beta \geq 0$ (that is, $V_\beta$ is a valid covariance matrix), the Schur complement in this term cannot be PSD (if it were it would imply that \erf{LMIL} were true when we have assumed the opposite).  Thus we have
\beq G=\Delta_{V_\beta+i\Sigma_\beta}=V_\alpha-C\left(V_\beta+i\Sigma_\beta\right)^{-1}C\tp \ngeq 0.
\eeq
Consider an eigenvector of $G$, $\vec{\nu}$, associated with a negative eigenvalue, that is, $G\vec{\nu}=-g\vec{\nu}$ where $g > 0$. We can choose a measurement $A$ along an axis such that $T^A$ shares the eigenvector $\vec{\nu}$ so that $T^A \vec{\nu}=t\vec{\nu}$.  Now it is possible to arrange the measurement such that $t < g$. This is because it is always possible to make one eigenvalue suitably small (the eigenvalue for the conjugate variable will become large). Now we have chosen the measurement such that $G+T^A$ must have a negative eigenvalue in the $\vec{\nu}$ direction. Hence, for this choice of measurement we have \beq V_\alpha+T^A-C\left(V_\beta+i\Sigma_\beta\right)^{-1}C\tp \label{SchurNeg}\ngeq 0.
\eeq  This shows that if \erf{LMIL} does not hold then there exists a measurement $A$ such that \erf{lemma3} does not hold, which by Lemma \ref{part3} implies that for this measurement there does not exist an ensemble defined by $U$ such that \erfs{LMI1}{LMI2} hold. 

However, from Lemma \ref{part2} we know that if $W$ is not steerable then there must exist an ensemble defined by $U$ such that these equations hold. %do not hold then it is not possible to define an ensemble $U$ to prevent steering. 
Since they do not hold for all $A$ when \erf{LMIL} is not true, we see that if \erf{LMIL} is not true then we cannot define a suitable ensemble $U$ to prevent steering. Therefore, a Gaussian state $W$ is \emph{not} steerable \emph{iff} \erf{LMIL} is true.
\end{proof}

Theorem 5 provides a sufficient condition for demonstrating that the state $W$ is steerable (by any measurements), and hence specifies an \emph{upper} bound on $\eta_{\rm steer}$. To illustrate this, it is useful to consider a simple example.

\subsubsection{Two mode states and the EPR paradox}
We now consider the simplest case where Alice and Bob share a Gaussian state $W$ in which they each have a single mode. It is well known that such a Gaussian state can be brought into standard form using local linear unitary Bogoliubov operations (LLUBOs), so that the CM takes the form \cite{DuaEtalPRL00}
\beq
V_{\alpha\beta}= \left(%
\begin{array}{cccc}
n & 0 & c & 0  \\
  0& n &0 & c' \\
c & 0 & m &0 \\
 0& c' &0& m \\
\end{array}
\right),\label{GaussStandardCovariance}
\eeq
where $n,m \geq 1$.  
%In this case we simply denote the mode operators as $\hat{q}_\alpha=\hat{a}+\hat{a}^\dagger$, $\hat{p}_\alpha=-i(\hat{a}-\hat{a}^\dagger)$ for Alice's modes and $\hat{q}_\beta=\hat{b}+\hat{b}^\dagger$, $\hat{p}_\beta=-i(\hat{b}-\hat{b}^\dagger)$ for Bob's. Here $\hat{a},\hat{b}$ and $\hat{a}^\dagger,\hat{b}^\dagger$ are annihilation and creation operators respectively for the two modes. The commutation relations can thus be simply expressed in the form of \erf{commutation}.

The Peres-Horodecki criterion for separability can be written as a linear matrix inequality for Gaussian states \cite{SimPRL00} as
\beq \tilde{V}_{\alpha\beta}+i\Sigma_{\alpha\beta} \geq 0,
\label{GaussSep}\eeq
where $\tilde{V}_{\alpha\beta}=\Lambda V_{\alpha\beta} \Lambda$; $\Lambda={\rm diag}(1,1,1,-1)$.
This can be determined by finding when the Schur complement (of the lower block) of $\tilde{V}_{\alpha\beta}+i\Sigma_{\alpha\beta}$ is PSD, which occurs only when
\beq
\left(m-\frac{c^2n}{n^2-1} \right)\left(m-\frac{{c'}^2n}{n^2-1} \right) \geq \left(1- \frac{cc'}{n^2-1} \right)^2. \label{GaussSepGeneral}
\eeq Hence two-mode Gaussian states defined by $V_{\alpha\beta}$ are separable \emph{iff} \erf{GaussSepGeneral} is satisfied.

%all the eigenvalues of the expression are non-negative. A straightforward calculation reveals that the eigenvalues are $\lambda_1=\gamma+\delta+1$, $\lambda_2=\gamma+\delta-1$, $\lambda_3=\gamma-\delta+1$, $\lambda_4=\gamma-\delta-1$. Clearly since $\gamma,\delta > 0$, the smallest eigenvalue is $\lambda_4$. Thus, \erf{GaussSep} will be violated when $\lambda_4 < 0$. This leads to the condition for states defined by \erf{GaussCovariance} to be \emph{nonseparable} 
%\beq
%\eta>\eta_{\rm ent}=\sqrt{\frac{\bar{n}}{1+\bar{n}}}.
%\label{GaussInsep}\eeq

For Gaussian states, which have a positive Wigner function, it is not possible to demonstrate violation of a Bell inequality using Gaussian measurements.  This is because the Wigner function gives an explicit hidden variable description which ensures satisfaction of Bell's inequality. %Hence for the measurements we consider, it is not possible to demonstrate violation of a Bell inequality for the Gaussian state $W_2$.

To determine if the state $W$ is steerable %the steerability boundary for these states 
it is a simple matter of testing if $V^{\alpha\beta}+\mathbf{0}_\alpha\oplus i\Sigma_\beta$ is PSD.  Again using Schur complements, we find that this is the case \emph{iff}
\beq
\left(m-\frac{c^2}{n} \right)\left(m-\frac{{c'}^2}{n} \right) \geq 1. \label{GaussSteerGeneral}
\eeq %and Gaussian states defined by $V^{\alpha\beta}$ are steerable by Gaussian measurements \emph{iff} this condition is \emph{not} satisfied.

%The eigenvalues of this matrix are given by
%\beq \lambda_i=\frac{2\gamma\pm 1\pm\sqrt{1+4\delta^2}}{2},
%\eeq where again since $\gamma,\delta > 0$ it is clear that the eigenvalue with both the negative signs is the smallest. Thus the state $W_2$ is steerable if this eigenvalue is negative. Testing this leads to an upper bound on the condition for steerability
%\beq\eta > \sqrt{\frac{1+2\bar{n}}{2(1+\bar{n})}}\geq \eta_{\rm steer}
%.\label{GaussSteerCond2}\eeq
%
%Note that this is only an upper bound as we have considered a restricted class of measurements. However, it is straightforward to verify (by comparing \erf{GaussInsep} and \erf{GaussSteerCond2}) that $\eta_{\rm steer} \geq \eta_{\rm ent}$ for all $\bar{n}$.  This is shown explicitly for some small values of $\bar{n}$ in Fig. \ref{combinedfigures}(d).  Since it is not possible to demonstrate Bell nonlocality for a Gaussian state with Gaussian measurements we have also plotted an upper bound on $\eta_{\rm Bell}$ in Fig. \ref{combinedfigures}(d).

%\subsubsection{The EPR paradox}
Recall that the interest in, and even the name, \emph{steering} arose in response to the EPR paradox.  Therefore, one would expect that any reasonable characterization of steering should include the EPR paradox.  This is indeed the case for our formulation of steering.
For the class of two-mode Gaussian states that we have been considering,  Reid \cite{ReiPRA89} has argued that the  EPR ``paradox'' is demonstrated \emph{iff} the product of the conditional variances $V(q_\beta|q_\alpha)$ and
$V(p_\beta|p_\alpha)$  violates the uncertainty principle.  
This is the case if the conditional variances do \emph{not} satisfy
\beq
V(q_\beta|q_\alpha)V(p_\beta|p_\alpha)\geq 1.\label{HUP}
\eeq
For a general two-mode Gaussian state $W$ the conditional variances take the form,
\beq
V(q_\beta|q_\alpha)=\ {}_{\mu}^{\!\!\rm min}[(q_\beta-\mu q_\alpha)^2]=
%&=&{}_{\mu}^{\!\! \rm min}[m -2\mu c+\mu^2n]\nonumber\\
m-\frac{c^2}{n},
\eeq
and similarly for $V(p_\alpha|p_\beta)$. Thus \erf{HUP} is exactly \erf{GaussSteerGeneral}.
 That is, the EPR ``paradox'' occurs precisely when $W$ is 
 steerable with Gaussian measurements.  This example confirms
that the EPR ``paradox'' is merely a particular case of steering. As is well known \cite{BowEtalPRA04}, Reid's EPR condition is strictly stronger than the condition for nonseparability \erf{GaussSepGeneral}. The fact that the EPR ``paradox'' is an example of steering explains why the EPR condition is stronger than nonseparability; as we have shown in previous examples steering is a  strictly stronger concept than nonseparability.

\subsubsection{Symmetric two-mode states}\label{GaussSym}

Finally we consider the specific case of 
%Ideally we would calculate the entanglement properties of this general two-mode Gaussian state.  However, the analytical expression for steerability in this general case is difficult to work with. Instead we note that the covariance matrix for 
two-mode Gaussian states prepared by optical parametric amplifiers \cite{BowEtalPRA04}. When the entanglement is symmetric between the two modes the covariance matrix describing such states has a particularly simple form.  The continuous variable entanglement properties of such a state has recently been characterized experimentally \cite{BowEtalPRA04}. In this case the covariance matrix of the state $W$ has just two parameters, $\eta$ and $\bar{n}$: 
\beq
{\rm CM}[W_{\bar n}^\eta]=V_2^{\alpha\beta}= \left(%
\begin{array}{cccc}
 \gamma & 0 & \delta& 0  \\
  0&\gamma &0 & -\delta\\
 \delta&0&\gamma&0 \\
 0&-\delta&0&\gamma \\
\end{array}
\right),\label{GaussCovariance}
\eeq
where $\gamma=1+2\bar{n}$ and $\delta = 2\eta\sqrt{\bar{n}(1+\bar{n})}$. 
Here  $\bar{n}$ is the mean photon number for each party, and $\eta$ is a mixing parameter
defined analogously with the other examples except that here it is the covariance
matrix that is linear in $\eta$, not the state matrix.%the parameter $\eta$ is a mixing parameter defined analogously with the other examples.  
%The mode operators are $q_\alpha=\hat{a}+\hat{a}^\dagger$, $p_\alpha=-i(\hat{a}-\hat{a}^\dagger)$ for Alice's modes and $q_\beta=\hat{b}+\hat{b}^\dagger$, $p_\beta=-i(\hat{b}-\hat{b}^\dagger)$ for Bob's. $\hat{a},\hat{b}$ and $\hat{a}^\dagger,\hat{b}^\dagger$ are annihilation and creation operators respectively for the two modes. For convenience we define $R=(q_\alpha,p_\alpha,q_\beta,p_\beta)$. The commutation relations for the mode operators can thus be compactly expressed as
%\beq [R_i,R_j]=2i\Sigma^{ij}_{\alpha\beta},\ \ \ \ \ \ \ \ i,j=1,2,3,4 \label{RR}
%\eeq
%where $\Sigma_{\alpha\beta}$ is now a $4\times4$ matrix as defined in \erf{SigmaAB} and $\Sigma^{ij}_{\alpha\beta}$ denotes a matrix element of $\Sigma_{\alpha\beta}$.

For such symmetric states the separability condition, \erf{GaussSepGeneral}, becomes
\beq
\left(\gamma-\frac{\delta^2\gamma}{\gamma^2-1} \right)^2 \geq \left(1+ \frac{\delta^2}{\gamma^2-1} \right)^2.
\eeq
Substituting for the values of $\gamma$ and $\delta$ we find that the condition for states defined by \erf{GaussCovariance} to be \emph{nonseparable} is simply
\beq
\eta>\eta_{\rm ent}=\sqrt{\frac{\bar{n}}{1+\bar{n}}}.
\label{GaussInsep}\eeq

%For Gaussian states with a positive Wigner function it is not possible to demonstrate violation of a Bell inequality using Gaussian measurements.  This is because the Wigner function acts as an explicit hidden variable description which ensures satisfaction of Bell's inequality. Hence for the measurements we consider, it is not possible to demonstrate violation of a Bell inequality for the Gaussian state $W_2$.

In determining when symmetric two-mode Gaussian states are steerable, we find that \erf{GaussSteerGeneral} becomes
\beq
\left(\gamma-\frac{\delta^2}{\gamma} \right)^2  \geq 1.
\eeq
%
%To determine if the state $W_2$ is steerable %the steerability boundary for these states 
%it is a simple matter of testing if $V_2^{\alpha\beta}+\mathbf{0}_\alpha\oplus i\Sigma_\beta$ is PSD. The eigenvalues of this matrix are given by
%\beq \lambda_i=\frac{2\gamma\pm 1\pm\sqrt{1+4\delta^2}}{2},
%\eeq where again since $\gamma,\delta > 0$ it is clear that the eigenvalue with both the negative signs is the smallest. Thus the state $W_2$ is steerable if this eigenvalue is negative. Testing this leads 
Hence, as an upper bound on the condition for steerability we have
\beq\eta > \sqrt{\frac{1+2\bar{n}}{2(1+\bar{n})}}\geq \eta_{\rm steer}
.\label{GaussSteerCond2}\eeq
This is an upper bound as we have only considered a restricted class of measurements. These results are plotted  for some small values of $\bar{n}$ in Fig. \ref{combinedfigures}(d).  Since it is not possible to demonstrate Bell nonlocality for a Gaussian state with Gaussian measurements we have also plotted an upper bound on $\eta_{\rm Bell}$ in Fig. \ref{combinedfigures}(d).

\section{Discussion\label{sec:disc}}
We have introduced a rigorous formulation of the concept of steering and given a number of examples to demonstrate where this concept fits in the hierarchy of entangled states. Both our operational and mathematical formulations of steering 
%in terms of a task 
leads to the notion that steerable states lie between nonseparable states and those entangled states which violate a Bell inequality.  In particular our example for $2\times 2$ Werner states establish that this is a \emph{strict} hierarchy.  Our other examples are consistent with this fact.

Recently there has been renewed interest in classifying the resources present in quantum states. For instance, it has been proposed that nonlocality itself is a separate resource from entanglement (see \cite{BruEtalNJP05} and references therein).  This has been motivated by the fact that for suggested measures of nonlocality, the maximally nonlocal states are not necessarily maximally entangled states.  Our work provides an interesting addition to the increasingly complex task of characterizing quantum resources. Clearly steerability is another form of nonlocality that a quantum state may possess.

The nonlocality of entangled states has also recently been studied in the context of robustness to noise.  Ref. \cite{AciEtalPRL07} determines the maximum amount of noise that an arbitrary bipartite state can accept before its nonlocal correlations  (i.e. its ability to violate a Bell inequality) are completely ``washed out".  They do this by determining when the resulting state's correlations can be explained by a ``local model".  In fact, the local models defined in Ref. \cite{AciEtalPRL07} correspond to LHS models for Bob in our terminology. That is, as they recognize \cite{AciEtalPRL07}, the concept of steering is useful for proving new bounds for Bell-nonlocality, since the latter is strictly stronger.

The inherent asymmetry in the definition of steerability may suggest applications for asymmetric entangled states.  It may appear that a link exists between the recently proposed asymmetric measures of entanglement \cite{HorEtalARX05} and steerability.  While conceptually appealing, this seems unlikely as states with asymmetric entanglement as defined in Ref. \cite{HorEtalARX05} 
necessarily contain bound entanglement.  A connection between steerable states and bound entangled states is unlikely, as we have shown that steerable states exist for $d=2$ (and no bound entangled states exist for $d=2$).

There remain a number of open questions relating to steerability.  We have demonstrated the link between the EPR paradox and steerability for two-mode Gaussian states.  The EPR paradox has been demonstrated experimentally for Gaussian states, however it is difficult to prepare an EPR type experiment for other quantum states. This raises the question: might tests of steerability provide experimental evidence for EPR type correlations in non-optical experimental implementations?

From an experimental perspective, the question as to whether it is possible to define a steerability witnesses or operators (in analogy with entanglement witnesses and Bell operators) is particularly appealing.  This would provide a straightforward experimental test for determining if a given state is steerable.  Such a test would simultaneously demonstrate that the given state is entangled.  This will be addressed in future work.

{ Finally, our operational definition of steering in terms of a task involving exchanges of quantum systems with an untrusted partner is reminiscent of the scenarios common in quantum complexity theory such as interactive proof systems and other kinds of quantum games (see for example~\cite{GutWatARX07}). We do not know of a direct way of mapping steering as we have defined it here onto these problems but it is interesting to ask if steering may play a role in some way comparable to Bell inequality violation in \cite{CleEtalARX04} for example. Secondly is it possible to define some useful quantum protocol for which the class of steerable states is useful; that is, is there a task for which nonseparable states are an insufficient resource, but steerable states allow the protocol to be implemented?}
%Finally, our operational definition of steering in terms of a task is reminiscent of Merlin-Arthur type scenarios common in computational complexity analyses. This raises two questions: first, does the class of steerable states fall into some known complexity class?
%\note{--Do you agree with this Andrew?  Could you possibly add a relevant reference and elaborate? Is there anything else you can think of that would be appropriate for the discussion section?-- Steve.} Secondly is it possible to define some useful quantum protocol for which the class of steerable states is useful; that is, is there a task for which nonseparable states are an insufficient resource, but steerable states allow the protocol to be implemented?

We conclude by commenting that we expect the answers to these questions (and others) to prove steering a useful concept in the context of quantum information science.

\acknowledgments{This work was supported by the Australian Research Council and the  State of Queensland.
We would like to thank Rob Spekkens, Volkher Scholz, Antonio Ac\`in, and Michael Hall for useful discussions.}

\appendix

\section{EPR-correlations, entanglement, and steering: a history of terms}\label{App:history}

As stated in the Introduction, although Werner's 1989 paper \cite{WerPRA89} is often cited as introducing the dichotomy of entangled versus separable states,  it is 
important (for the discussion in this Appendix) to note that he used neither the 
term entangled nor the term separable. These seem to have not been used in their presently accepted sense until 1996, by Bennett \ea\ \cite{BenEtalPRA96}, and Peres \cite{PerPRL96} respectively. Rather, he used the terms 
``EPR-correlated states'' versus ``classically correlated states''.   His main result, restated  
in these terms, was that some ``EPR-correlated states'' %admit a description by LHVs and so 
conform with Bell's concept of ``locality''.

Our recent work \cite{WisEtalPRL07} also considered the issue of mixed states and EPR correlations.  Specifically, we rigorously defined the class of states that can be used to demonstrate the nonlocal effect which EPR identified in 1935. 
Contrary to Werner's {\em terminology}, we established that the set of such EPR-correlated states is {\em not} by definition complementary to the set of locally preparable states. Using this concept of EPR-correlated states, the main results of our paper can, ironically, be expressed 
entirely in statements contradicting Werner's {\em natural-language descriptions} of his results.  
First: It is true (as Werner states) that some EPR-correlated states respect Bell-locality, but, contrary to his (natural-language) claims, Werner did {\em not} prove this. Our proof \cite{WisEtalPRL07} of this fact makes use of Werner's result, but also requires the much more recent result of 
Acin \ea\ \cite{AciEtalPRA06}. Second: what Werner's result {\em actually} proves, contrary to his stated dichotomy, is that some states that are not separable (classically correlated)   
are nevertheless {\em not} EPR-correlated (from which one can conclude also that they are Bell-local). To summarize: we used Werner's result to help 
prove that the set of Bell-nonlocal states is a {\em strict} subset
of the set of EPR-correlated states, which in turn is a {\em strict} subset of the set of nonseparable states.

We emphasize that we are not disputing at all the mathematical validity of Werner's result, nor its importance, 
nor his understanding of it. We dispute only his use of the term ``EPR-correlated
states'' to refer to nonseparable states, which he says ``is to emphasize the crucial role of such states in the 
Einstein-Podolsky-Rosen paradox and for the violations of Bell's inequalities''.   This explanation for the name 
could equally be used to justify calling nonseparable states ``Bell-correlated states'', but that would be nonsensical 
since the point of Werner's paper is that, in the mixed-state case, not all nonseparable states can exhibit correlations that violate a Bell's inequality.
Similarly, we maintain that if the term ``EPR-correlated states'' were to be applied to mixed states, then it should be reserved for those 
states for which the correlations can actually be used to demonstrate
the EPR paradox. Prior to our Letter, 
no rigorous and general definition of this paradox had been given, and so no good definition of 
``EPR-correlated states'' existed. Giving such a definition is no mere semantic exercise; 
as stated in the preceding paragraph, our work identifies this as a new class of quantum 
states, distinct both from the Bell-nonlocal ones and the nonseparable ones.

During %in 
the past decade \sch's term ``entangled states'' has replaced Werner's term ``EPR-correlated states'' 
(which he credited to Primas \cite{Pri84}) as a synonym for nonseparable states. Nevertheless, there is still potential for confusion if we were to 
promote the term ``EPR-correlated'' for the new class of states we defined and 
categorized in Ref.~\cite{WisEtalPRL07}. For that reason we proposed instead the term ``steerable'' for this class of states, a term that has been used increasingly in recent years
\cite{VujHerJPA88,VerPhD02,CliEtalFoP03,SpePRA07,KirFPL06}.

\section{Isotropic state steering}\label{GaussApp}

\subsection{Optimal ensemble}\label{App:optens}
First choose an orthonormal basis $\ket{1},\ket{2},...,\ket{d}$
to describe the uniform ensemble $F^\star$. Then consider randomly
generated unnormalized states \beq \ket{\tilde{\psi}}=\frac{1}{\sqrt{d}}\sum_{j=1}^d
z_j\ket{\phi_j}\label{psirandom},\eeq where $z_j$ are zero-mean Gaussian random variables with the properties $\langle
z^*_jz_k\rangle=\delta_{j,k}$ and $\langle z_jz_k\rangle=0$. Writing $\ket{\tilde{\psi}}=m\ket{\psi}$, we denote the measure for this ensemble as $d\mu_G(\psi,m)$.  From \erf{psirandom} it is straightforward to see \beq
\hat{U}\ket{\tilde{\psi}}=\frac{1}{\sqrt{d}}\sum_{j,j'=1}^d
z_jU_{jj'}\ket{\phi_{j'}}=\frac{1}{\sqrt{d}}\sum_{j'=1}^d
w_{j'}\ket{\phi_{j'}},\label{isounitarity}\eeq where $w_{j'}=\sum
U_{jj'}z_j$. Due to unitarity, $\langle
w_{j'}^*w_{k'}\rangle=\delta_{j'k'}$ and $\langle
w_{j'}w_{k'}\rangle=0$ (that is, the $w$s satisfy the same
statistical relations as the $z$s). Hence \beq
d\mu_G(\psi,m)=d\mu_G(\hat{U}\psi,m),\ \forall \hat{U},\eeq  which means that the measure factorizes into a constant measure over $\psi$
(the Haar measure) and a measure over the weightings $m$, and can be written as \beq
d\mu_G(\psi,m)=d\mu_{\rm Haar}(\psi)d\mu_G(m).\eeq Hence instead of simply using the Haar measure $d\mu_{\rm Haar}$ to describe the distribution of the ensemble $F^\star$ we may use the Gaussian measure $d\mu_G(\psi,m)$.

For simplicity we go on to define $z=\sqrt{u}e^{i\theta}$ so that $d\mu_G(\psi,m)$ becomes
\bqa \wp(u_1,...,u_d,\theta_1,...,\theta_d)du_1...du_d d\theta_1...d\theta_d \ \ \ \ \ \ \ \ \nonumber\\
 =\frac{1}{(2\pi)^d} \exp\left(-\sum_{i=1}^d u_i\right) du_1...du_d d\theta_1...d\theta_d
\eqa which is normalized:
\beq
%\hspace{-1cm}\int du_1...du_d d\theta_1...d\theta_d \wp(u_1,...,u_d,\theta_1,...,\theta_d)\hspace{-7.25cm} \nonumber\\&= & 
\frac{1}{(2\pi)^d} \int_0^\infty du_1...du_d\exp\left(-\sum_{i=1}^d u_i\right)\int_0^{2\pi} d\theta_1...d\theta_d =1.%  \nonumber\\
%&=&1.%\hspace{7.3cm} 
\eeq

\subsection{Normalization term}\label{App:norm}

The denominator of \erf{IsoAv} evaluates to
\bqa \int d\mu(m)m^2&=& \int\int d\mu(m)m^2 d\mu_{\rm Haar}(\psi)\nonumber\\
&=&\int d\mu_G(\psi,m)\langle\tilde{\psi}|\tilde{\psi}\rangle\nonumber\\
&=&\int d\mu_G(\psi,m)\sum_{i,j=1}^d\frac{z_i^\ast z_j}{d}\bra{\phi_i}\phi_j\rangle\nonumber\\
&=&\int d\mu_G(\psi,m)\sum_{i=1}^d\frac{|z_i|^2}{d}\nonumber\\
&=&{\rm E}_G\left[\sum_{i=1}^d\frac{|z_i|^2}{d}\right]=1.
\eqa
We have used the facts that $\int d\mu_{\rm Haar}(\psi)=1$ and $m^2=\langle\tilde{\psi}|\tilde{\psi}\rangle$, while ${\rm E}_G[x]$ denotes the expected value of $x$ with respect to the Gaussian measure $d\mu_G(\psi,m)$.

\subsection{Evaluating \erf{H_d}}\label{App:Int}

To calculate the integral in \erf{H_d} we need the limits of integration.  These are determined by 
considering the states in Hilbert space which are closer to
$\ket{a}\bra{a}$ than any other basis state as outlined in \erf{IsoMethod}. We can choose the orthornormal basis in \erf{psirandom} such that $\ket{a}$ is one of the basis states.    Since the states $\ket{\tilde{\psi}}$ are unnormalized we must perform the the integral relating to $\ket{a}$ over the complete range from
0 to $\infty$. That is, the coefficient $u_a$ associated with $\ket{a}$ ranges from 0 to $\infty$. However, since this must be the largest parameter, the
other $u_i$ must only range from 0 to $u_a$. Using these limits the integral $\int_ad\mu_G(\psi,m)|\bra{a}\tilde{\psi}\rangle|^2$ becomes
\bqa &&\!\!\!\!\!\!\!\!\!\!\!\!\int_ad\mu_G(\psi,m)\left|\bra{a}\frac{1}{\sqrt{d}}\sum_{j=1}^dz_j\ket{\phi_j}\right|^2\nonumber\\
&=&\!\!\frac{1}{d}\int_ad\mu_G(\psi,m)|z_a|^2\nonumber\\
&=&\!\!\frac{1}{d}\int_ad\mu_G(\psi,m)u_a\nonumber\\
&=&\frac{1}{d(2\pi)^d}\int_0^\infty du_a u_a\int_0^{u_a}du_2 ...
\int_0^{u_a}du_d\nonumber\\
&&\phantom{\frac{1}{d(2\pi)^d}\int}\times\int_0^{2\pi}d\theta_1...\int_0^{2\pi}d\theta_d \exp\left(-\sum_{i=1}^d
u_i\right)\nonumber\\
&=&\frac{(2\pi)^d}{d(2\pi)^d}\int_0^\infty du_au_ae^{-u_a}
\left(\int_0^{u_a}due^{-u}\right)^{d-1}\nonumber\\
&=&\frac{1}{d}\int_0^\infty du_au_ae^{-u_a}
\left(1-e^{-u_a}\right)^{d-1}\nonumber\\
&=&\frac{1}{d}\int_0^\infty du_au_ae^{-u_a}
\sum_{k=0}^{d-1}\binom{d-1}{k}\left(-e^{-u_a}\right)^k\nonumber\\
&=&\frac{1}{d}\sum_{k=0}^{d-1}(-1)^k\binom{d-1}{k}\left[\frac{1}{(k+1)^2}\right]\nonumber\\
&=&\frac{1}{d}\sum_{k=1}^{d}(-1)^{k-1}\binom{d-1}{k-1}\frac{1}{k^2}\nonumber\\
&=&\!\!\frac{1}{d^2}\sum_{k=1}^d\frac{(-1)^{k-1}}{k}\binom{d}{k}\equiv \epsilon_d\label{Appeps_d}.\eqa
It is possible to further simplify $\epsilon_d$. This can be done by considering the following expression %$(1/d^2)\int_0^1dx[1-(1-x)^d]/x$ which can be evaluated as 
\bqa \frac{1}{d^2}\int_0^1dx\frac{1-(1-x)^d}{x} &=&\frac{1}{d^2}\int_0^1dx\frac{1-\sum_{k=0}^d
\binom{d}{k}(-x)^k}{x}\nonumber\\
&=&\frac{1}{d^2}\sum_{k=1}^d
\frac{(-1)^{k+1}}{k}\binom{d}{k}.
\eqa 
It is not immediately obvious that the above integral is a simpler expression for $\epsilon_d$.  However, this expression can be evaluated alternatively using the subsitution $y=1-x$, which gives
\bqa \epsilon_d &=&\frac{1}{d^2}\int_0^1dy\frac{1-y^d}{1-y}\nonumber\\
&=&\frac{1}{d^2}\int_0^1dy(1-y^d)\sum_{k=0}^\infty y^k\nonumber\\
&=&\frac{1}{d^2}\left(\sum_{k=0}^{d-1}\frac{1}{k+1}+ \sum_{k=d}^\infty\frac{1}{k+1}-\sum_{k=0}^\infty\frac{1}{k+d+1}\right)\nonumber\\
&=&\frac{1}{d^2}\sum_{k=1}^{d}\frac{1}{k},
\eqa which is the result in \erf{H_d}.%where $H_d=1+1/2+1/3+\ldots+1/d$ is the Harmonic series.

\section{Inept state steering}

\subsection{Optimal ensemble}\label{IneptOptimalApp}

We need to show that the conditions for Lemma 1 hold for the ensemble defined by \erf{ineptoptimal}.  That is, we need to show that \erf{ineptoptimal} defines an optimal ensemble.  In this instance $G$ is  the group generated by $(1/2)\hat{\sigma}_z\otimes\mathbf{I}-(1/2)\mathbf{I}\otimes\hat{\sigma}_z$, so $g\rightarrow \phi \in \left[ 0,2\pi\right)$ and 
\beq
\hat{U}_{\alpha\beta}(\phi) = \exp\left[-i\phi\hat{\sigma}_z/2\right]\otimes\exp\left[i\phi\hat{\sigma}_z/2\right].
\eeq
For the particular measurement strategy chosen we need to consider only two types of measurement $\hat{\sigma}_z$ and $\hat{\sigma}_\theta$. The condition $\hat{U}_\alpha\dg(\phi)\hat{A}\hat{U}_\alpha(\phi) \in {\mathfrak M}_\alpha$ clearly holds for $\hat{A}=\hat{\sigma}_z$ since 
\beq \exp\left(i\phi\hat{\sigma}_z\right)\hat{\sigma}_z\exp\left(-i\phi\hat{\sigma}_z\right)=\hat{\sigma}_z.\eeq Therefore \erf{lem1} holds trivially. 

Now it simply remains to test these conditions for measurements of the form of $\hat{A}=\hat{\sigma}_\theta$.  In this case we have 
\bqa
\hat{U}_\alpha\dg(\phi)\hat\sigma_\theta\hat{U}_\alpha(\phi)
&=& \exp\left(i\phi\hat{\sigma}_z\right) \hat{\sigma}_\theta
 \exp\left(-i\phi\hat{\sigma}_z\right) \nonumber\\
 &=& \cos\left(\theta-\phi \right)\hat{\sigma}_x+\sin\left(\theta-\phi \right)\hat{\sigma}_y\nonumber \\
 &=&  \hat{\sigma}_{\theta-\phi},
 \eqa
which is in $\mathfrak{M}_\alpha$.
%\bqa \exp\left(i\phi\hat{\sigma}_z\right)&\hat{\sigma}_\theta& \exp\left(-i\phi\hat{\sigma}_z\right)\nonumber\\&=&\cos\left(\theta-\phi \right)\hat{\sigma}_x+\sin\left(\theta-\phi \right)\hat{\sigma}_y\nonumber\\
%&=& \hat{\sigma}_{\theta-\phi}.\eqa
%\hat{\sigma}_{-\theta}\in \mathfrak{M}_\alpha.\eqa  
%That is, we simply have $\theta\rightarrow \theta-\phi$ and hence $\hat{\sigma}_{\theta-\phi}\in \mathfrak{M}_\alpha$. 
Thus to test \erf{lem1} we need to evaluate
$\tilde{\rho}_a^{\sigma_{\theta-\phi}}$. Using \erf{sigphi} we can see that this simply evaluates to
\bqa
\tilde{\rho}_a^{\sigma_{\theta-\phi}}&=&\frac{1}{2}\left[\mathbf{I}+a \eta\sqrt{\epsilon(1-\epsilon)}\cos(\theta-\phi)\hat{\sigma}_x \right.\nonumber\\
&&\phantom{\frac{1}{2}}\left.+a \eta\sqrt{\epsilon(1-\epsilon)}\sin(\theta-\phi)\hat{\sigma}_y 
- (1-2\epsilon)\hat{\sigma}_z \right]. \nonumber\\ \eqa
Finally we evaluate
\bqa
\hat{U}_\beta(\phi) \tilde{\rho}_a^{\sigma_\theta} \hat{U}_\beta\dg(\phi)
&=& \exp\left(i\phi\hat{\sigma}_z\right)\tilde{\rho}_a^{\sigma_\theta}\exp\left(-i\phi\hat{\sigma}_z\right)\nonumber\\
&=&\frac{1}{2}\left(\begin{matrix}e^{-i(\theta-\phi)}
2\epsilon &   \kappa e^{-i(\theta-\phi)} \\
\kappa e^{i(\theta-\phi)}  & 2(1-\epsilon )
\end{matrix} \right)\nonumber\\
&=& \tilde{\rho}_a^{\sigma_{\theta-\phi}},\eqa
%\hat{\sigma}_{-\theta}\in \mathfrak{M}_\alpha.\eqa  
where $\kappa=a\eta\sqrt{\epsilon(1-\epsilon)}$. Thus \erf{lem1} also holds for measurements of $\hat{\sigma}_\theta$. Hence, the conditions of Lemma 1 hold  and the ensemble defined by \erf{ineptoptimal} is of the form of the optimal ensemble.

\subsection{Steering bound for $\epsilon=1/2$}\label{App:IneptHalf}

We know that for $d=2$ the Werner and isotropic states are equivalent. Now consider the inept states when $\epsilon=1/2$. In this case \erf{PESmatrix} becomes
\beq W_{\frac{1}{2}}^\eta=\eta\ket{\psi}\bra{\psi}+(1-\eta)\frac{\bf{I}}{4},
\eeq where $\ket{\psi}=\frac{1}{\sqrt{2}}\left(\ket{0_\alpha 0_\beta}+\ket{1_\alpha 1_\beta}\right)$.  Comparing this with \erf{isotropicstates} for isotropic states (when $d=2$), one immediately sees that the expressions are identical.  Hence, for $\epsilon=1/2$ the inept states are equivalent to the $d=2$ isotropic (and Werner) states.

Setting $\epsilon=1/2$ in \erf{PESsteeringcond2} we find an upper bound of $0.5468$ on $\eta_{\rm steer}$. However, we know that the steering boundary for $d=2$ isotropic states occurs at $\eta=1/2$. Thus for $\epsilon=1/2$ we can do better than an upper bound on $\eta_{\rm steer}$ for inept states; due to the equivalence with isotropic states we know that the true $\eta_{\rm steer}$ occurs at $\eta=1/2$.  We plot this as a separate point at $\epsilon=1/2$ in Fig. \ref{combinedfigures} (c).

\section{Gaussian state steering}

\subsection{Proof of Lemma 2}\label{lemma2proof}

First, suppose \erf{LMIL} is true. Thus the matrix $V_{\alpha\beta} + {\bf 0}_\alpha\oplus i\Sigma_\beta$ is PSD. Now since \erf{LMIL} is assumed true, and we know that $V_\alpha\geq 0$, taking the Schur complement of $V_\alpha$ in \erf{LMIL} we see
\beq V_\beta+i\Sigma_\beta - C\tp V_\alpha^{-1}C  \geq 0,
\eeq which implies \erf{LMI1} where $U=V_\beta - C\tp V_\alpha^{-1}C$. This LMI allows us to define an ensemble $F^U=\{\rho_\xi^U\wp^U_\xi\}$ of Gaussian states with CM$[\rho^U_\xi]=U$, distinguished by their mean vectors ($\xi$). 
  
Now we wish to see if the ensemble $U$ defined above could be used to simulate Bob's conditioned state $V_\beta^A$. This will be the case \emph{iff} $V_\beta^A-U$ is PSD as explained below. Evaluating this matrix we see that 
\bqa V_\beta^A-U&=&V_\beta -  C\tp(V_\alpha+T^A)^{-1}C-V_\beta +  C\tp V_\alpha^{-1}C\nonumber\\
&=&C\tp\left[V_\alpha^{-1}-(V_\alpha+T^A)^{-1}\right]C.
\eqa Both $C$ and $C\tp$ are positive matrices, so the above expression is PSD if and only if the bracketed term is PSD.  To prove this is so, we make use of the Woodbury formula, which can be expressed as
\beq X^{-1}-\left(X+YZ\tp\right)^{-1}=X^{-1}Y\left(I+Z\tp X^{-1}Y\right)^{-1}Z\tp X^{-1}\label{Woodbury}.
\eeq Thus to check the positivity of $V_\alpha^{-1}-(V_\alpha+T^A)^{-1}$ we set $X=V_\alpha$, $Y=\sqrt{T^A}$ and $Z\tp=\sqrt{T^A}$, and thus
\bqa V_\alpha^{-1}&-&(T^A+V_\alpha)^{-1}\nonumber\\%&=&V_\alpha^{-1}-\left(V_\alpha^{-1}-\left[V_\alpha^{-1}\sqrt{T^A}\{I+\sqrt{T^A}V_\alpha^{-1}\sqrt{T^A}\}^{-1}\sqrt{T^A}V_\alpha^{-1}\right]\right)\nonumber\\
&=&V_\alpha^{-1}\sqrt{T^A}\left(I+\sqrt{T^A}V_\alpha^{-1}\sqrt{T^A}\right)^{-1}\sqrt{T^A}V_\alpha^{-1}.\nonumber\\
\eqa
 Now the covariance matrices $V_\alpha$ and $T^A$ are positive by definition, so their inverse and square root respectively must also be positive matrices.  Since the product and the sum of two positive matrices is PSD, the above expression is PSD if and only if $\sqrt{T^A}V_\alpha^{-1}\sqrt{T^A}$ is PSD, which holds since any matrix $ABA^T$ is PSD whenever B is PSD.
%Now the covariance matrix $V_\alpha$ is positive by definition, so its inverse $V_\alpha^{-1}$ must also be a positive matrix.   Similarly, the covariance matrix $T^A$ is positive, so its square root $\sqrt{T^A}$ must be positive. So the above expression is PSD if and only if the bracketed term is. 
%%This term is the sum of identity (which is a positive matrix) and $\sqrt{T^A}V_\alpha^{-1}\sqrt{T^A}$ which is a product of positive matrices (as we just showed).  
%Since the product of positive matrices is a positive matrix and the sum of two positive matrices is PSD, $I+\sqrt{T^A}V_\alpha^{-1}\sqrt{T^A}$ is PSD. 
Thus the ensemble defined by $U=V_\beta - C\tp V_\alpha^{-1}C$ satisfies \erf{LMI2}
%\beq V_\beta^A-U\geq 0 \label{LMI2},
%\eeq 
which implies that $\forall A\in\mathfrak{G},\ \rho^A_a$ is a Gaussian mixture (over $\xi$) of Gaussian states $\rho_\xi^U$, all with the same covariance matrix $U$, but with different mean vectors $\xi$. Specifically, $\wp(a|\xi,A)$ is a Gaussian distribution in $\xi$ with a  
mean vector equal to $a$ (which is determined by Alice's measurement $A$ and the bipartite Gaussian state $W$), and a covariance matrix equal to $V_\beta^A-U$.  As long as 
$V_\beta^A-U \geq 0$, this distribution is well defined, so that Bob's state $\rho_a^A$ is consistent with Alice merely sending Bob Gaussian states drawn from an ensemble $F^U=\{\rho_\xi^U\wp^U_\xi\}$  in which all states have a 
CM equal to $U$, and with mean vectors $\xi$ having a Gaussian distribution $\wp^U_\xi$ which has a covariance matrix $V_\beta - U = C\tp V_\alpha^{-1}C \geq 0$. Therefore $W$ is not steerable by Alice for all measurements $A\in\mathfrak{G}$ if \erf{LMIL} is true. \rule{0.5em}{0.5em} %Thus, if ${\bf L}$ is satisfied then steering is not possible.

\subsection{Proof of Lemma 3}\label{lemma3proof}

Suppose that $W$ is not steerable. This means that there is some ensemble $F=\{\wp_\xi\rho_\xi\}$ which satisfies \erf{steering2}. Therefore we know that Bob's conditioned state can be written as 
\beq \rho_a^A=\frac{\sum_\xi \wp_\xi \wp(a|A,\xi)\rho_\xi}{\sum_\xi \wp_\xi\wp(a|A,\xi)}\label{rhoA}.
\eeq 
This means that the covariance matrix satisfies
\beq {\rm CM}\left[\rho_a^A\right]=V_\beta^A\geq \frac{\sum_\xi \wp_\xi \wp(a|A,\xi){\rm CM}\left[\rho_\xi\right]}{\sum_\xi \wp_\xi\wp(a|A,\xi)}\label{CovrhoA}, 
\eeq since the CM of a state equal to a weighted sum of states must be at least as great as the weighted sum of the individual CMs. The equality occurs \emph{iff} all the means are the same.
Rearranging and taking a sum over $a$ on both sides gives
\beq  \sum_{\xi,a} \wp_\xi\wp(a|A,\xi)V_\beta^A\geq\sum_{\xi,a} \wp_\xi \wp(a|A,\xi){\rm CM}\left[\rho_\xi\right].\label{U}
\eeq
From the fact that $V_\beta^A$ is  independent of $a$, and using the facts that $\sum_a \wp(a|A,\xi)=1$ and $\sum_\xi \wp_\xi=1$, one sees that \erf{U} simplifies to
\beq V_\beta^A \geq \sum_{\xi} \wp_\xi {\rm CM}\left[\rho_\xi\right].\label{LMI2a}
\eeq
Defining $U=\sum_{\xi} \wp_\xi {\rm CM}\left[\rho_\xi\right]$ satisfies \erf{LMI1} by definition and \erf{LMI2a} implies \erf{LMI2}. Therefore if $W$ is not steerable then there exists an ensemble $U$ such that \erfs{LMI1}{LMI2} are true. \rule{0.5em}{0.5em}\\

\subsection{Proof of Lemma 4}\label{lemma4proof}

\erf{LMI2} defines the Schur complement of $V_\alpha+T^A$ in the following matrix
\beq M=\left(\begin{matrix}
V_\alpha+T^A &   C \\
C\tp  & V_\beta-U
\end{matrix} \right).
\eeq
Therefore, since $V_\alpha+T^A\geq 0$ (recall that we are considering Gaussian measurements), \erf{LMI2} is equivalent to the condition that the matrix $M$ be PSD. Now we know that the sum of two PSD matrices is PSD, so if $M\geq 0$ and using \erf{LMI1} we arrive at %$M+0_\alpha\oplus\left(U+i\Sigma\right)\geq 0$
\beq M+0_\alpha\oplus(U+i\Sigma_\beta)=\left(\begin{matrix}
V_\alpha+T^A &   C \\
C\tp  & V_\beta+i\Sigma
\end{matrix} \right)
\geq 0,
\eeq as an equivalent condition to \erfs{LMI1}{LMI2}. Finally, we know that $V_\beta+i\Sigma_\beta \geq 0$, so the Schur complement of this term in the above matrix must be PSD. \rule{0.5em}{0.5em}

\bibliography{SteeringPRA3}
\end{document}